\documentclass[onecolumn,prx,longbibliography,nofootinbib]{revtex4-2}

\usepackage{graphics} 
\usepackage{amsfonts,amssymb,amscd,amsmath,amsthm}
\usepackage{tikz}
\usetikzlibrary{quantikz}
\usepackage{enumerate}
\usepackage{epsfig}
\usepackage{pict2e}
\usepackage{keyval}
\usepackage{graphicx}
\usepackage{subfigure}
\usepackage{diagbox}
\usepackage{soul}
\usepackage[most]{tcolorbox}

\usepackage{physics}
\usepackage{xcolor}
\usepackage{hyperref}
\usepackage{pifont}

\newtheorem{theorem}{Theorem}
\newtheorem{lemma}{Lemma}
\newtheorem{corollary}{Corollary}

\newtheorem{definition}{Definition}

\newtcolorbox[auto counter]{mybox}[2][]{
	enhanced,
	breakable,
	colback=blue!5!white,
	colframe=blue!75!black,
	fonttitle=\bfseries,
	title=Box \thetcbcounter: #2,#1
}

\graphicspath{{./Figure/}}

\begin{document}

\title{Stream privacy amplification for quantum cryptography}

\author{Yizhi Huang}
\author{Xingjian Zhang}
\author{Xiongfeng Ma}
\email{xma@tsinghua.edu.cn}
\affiliation{Center for Quantum Information, Institute for Interdisciplinary Information Sciences, Tsinghua University, Beijing 100084, China}
	
\begin{abstract}
Privacy amplification is the key step to guarantee the security of quantum communication. The existing security proofs require accumulating a large number of raw key bits for privacy amplification. This is similar to block ciphers in classical cryptography that would delay the final key generation since an entire block must be accumulated before privacy amplification. Moreover, any leftover errors after information reconciliation would corrupt the entire block. By modifying the security proof based on quantum error correction, we develop a stream privacy amplification scheme, which resembles the classical stream cipher. This scheme can output the final key in a stream way, prevent error from spreading, and hence can put privacy amplification before information reconciliation. The stream scheme can also help to enhance the security of trusted-relay quantum networks and improve the practicality of randomness extraction for quantum random number generators.
\end{abstract}

\maketitle 

\section{INTRODUCTION}
As one of the first applications of quantum-information science, quantum key distribution (QKD) aims at establishing an information-theoretic secure key between two distant parties, Alice and Bob \cite{bennett1984proceedings,ekert1991quantum}. It applies fundamental laws of quantum physics to guarantee secure communication. The procedures of QKD can be divided into two parts: quantum operation and data postprocessing. Quantum operation includes the preparation, transmission, and measurement of quantum states for Alice and Bob to share raw key bits. The purpose of postprocessing is to extract an identical and private key from the raw data. This can be guaranteed by information reconciliation and privacy amplification, where the former guarantees the identity of key strings and the latter removes any potential information leakage to a possible eavesdropper, Eve~\cite{bennett1988privacy,maurer1993secret}.

Over the past three decades of development, QKD has experienced tremendous advancement from initial demonstrations in laboratories to practical implementation~\cite{xu2020secure}. In fiber, the communication distance has been pushed over 500 km \cite{chen2020sending,fang2020implementation}. Using quantum satellites, the communication distance has reached an intercontinental level \cite{liao2018satellite}. Researchers have also pushed QKD to a secret key rate of more than 10 Mbits/s~\cite{islam2017provably}. In addition to point-to-point linking, a number of field-test QKD networks have been conducted in many countries~\cite{elliott2005current,peev2009secoqc,stucki2011long, sasaki2011field,chen2009field,wang2010field}. In particular, China has recently successfully completed the 2000-km-long fiber-optic backbone link between Beijing and Shanghai~\cite{chen2021integrated}. Therefore, QKD is already a mature technique for real-life applications~\cite{qiu2014quantum}.
	
With the exciting developments on the experimental side, improving the practicality of QKD systems has become one of the essential issues in the field. Among all the stages in a QKD session, privacy amplification is one of the bottlenecks, which might still be technically difficult to implement in some realistic conditions. Existing privacy-amplification methods run as follows. After information reconciliation, Alice randomly chooses a hash function and sends it to Bob via a public classical channel. Both users hash their reconciled key strings with the hash function and obtain the final key. In practice, the family of Toeplitz-matrix hashing is widely adopted. Due to the matrix multiplication, privacy amplification can only process a block of reconciled key bits at a time.
Though recent security-analysis techniques have shown that the key rate can still be positive, with critical block sizes on the order of kilobits~\cite{Fung2010Finite}, smaller block sizes tend to make privacy amplification inefficient, resulting in lower key rates.
To guarantee the efficiency of privacy amplification, the block size is normally large in practice, typically on the order of megabits~\cite{Fung2010Finite}. Alice and Bob cannot perform privacy amplification until they accumulate an entire block of a reconciled key. This block feature of privacy amplification would cause unpleasant delays in some practical scenarios. For instance, in the satellite case, since the quantum signals can be transferred only when the ground station can ``see" the satellite with a clear atmosphere, it may take the satellite several orbits to accumulate enough data for one block of privacy amplification. Due to the unpredictable condition of the atmosphere, such a delay could take as long as days \cite{Liao2017Space}.

There are other cases where block privacy amplification could cause problems. If the ratio of the reconciled key to the final key is extremely large, the computational cost for the hash function can be very heavy. Such an issue has been encountered in randomness extraction of quantum random number generation (QRNG) as well, especially in the device-independent case \cite{liu2018device}. Due to the similarity between the definitions of randomness extraction and privacy amplification \cite{bennett1995generalized}, randomness extractors can be constructed by using universal hashing functions \cite{Ma2013Extractor,hayashi2016more} and have the same block feature as existing privacy-amplification schemes. The problem of the heavy computational cost is more serious in QRNG due to the larger amount of data and it restricts further improvement of the real-time generation speed.

Moreover, with block privacy amplification, the leftover error in information reconciliation would spread out to the entire block. Information reconciliation is normally done by bilateral error correction. For some error-correction schemes, there is a small probability of leaving some errors uncorrected. If Alice's and Bob's two strings are not exactly the same, the output strings from block privacy amplification will be totally different due to the universal property of the hash family.

Furthermore, existing quantum network implementations rely on trusted relays for key distribution \cite{chen2021integrated}, due to the limited transmission distance of point-to-point QKD links. Trusted QKD networks have been widely used in building intercity or backbone QKD communication links, such as the Hefei network \cite{Chen2021implementation} and the Beijing-Shanghai backbone link~\cite{chen2021integrated}. In the trusted-network scenario, all intermediate relays must be trusted because each of the relays can produce the final key. If one relay becomes compromised, the security of the whole network will be seriously threatened. In practice, it is challenging and expensive to guarantee high-level secure relays, which hinders the further commercialization and application of QKD. There are some attempts to reduce the dependence on trusted relays. Unfortunately, these solutions either require duplicate resources \cite{zhou2019security} or still assume that the intermediate relays do not attack the network intentionally \cite{lutkenhaus2016system}.

To address these issues, we reexamine the security proof for QKD based on quantum error correction \cite{lo1999unconditional}, where privacy amplification is reduced from phase-error correction \cite{shor2000simple}.
As a clear and simple showcase, we mainly focus on the Bennett-Brassard-1984 QKD protocol (BB84)~\cite{bennett1984proceedings} and go back to the original Lo-Chau security proof~\cite{lo1999unconditional}.
By rearranging the phase-error-correcting gates and error-syndrome measurement, we divide privacy amplification into two steps: (a) generate pseudorandom bits from a preshared key seed and a hash function; and (b) XOR the pseudorandom string from (a) and the reconciled key. We also prove that the hashing matrix in (a) can be reused. Then, Alice and Bob can generate pseudorandom bits offline. For real-time privacy amplification, they only need to perform the XOR operation in a bitwise manner. In the spirit of stream ciphers, the new scheme is conceptually different from the existing block privacy-amplification schemes. Such an essential difference guarantees the new scheme with the following practical features:
(1) it can output final key bits in a stream way; (2) it will not spread the errors of the input bit stings; and (3) it can be carried out ahead of information reconciliation.
 
The rest of this paper is organized as follows. In Sec.~\ref{Sc:preliminary}, we review QKD protocols and recap the security proof based on quantum error correction and its reduction to the prepare-and-measure case. In Sec.~\ref{Sc:newPA}, we reduce quantum phase-error correction to a new privacy-amplification procedure with a stream output and introduce its possible combination with delayed privacy amplification, classical cryptography, and QRNG. Finally, we conclude the paper and discuss possible future directions in Sec.~\ref{Sc:conclusion}.


\section{Preliminary}\label{Sc:preliminary}
\subsection{QKD protocols and security definition}
Here, we introduce the first and probably the most well-known QKD protocol, BB84 \cite{bennett1984proceedings}, and its entanglement version, Bennett-Brassard-Mermin-1992 (BBM92) \cite{bennett1992quantum}. Then, we show how their security can be established.

The procedures of the BB84 protocol are listed in Box \ref{box:BB84}. Alice and Bob have a quantum channel for state transmission and an authenticated classical channel for data postprocessing. In practice, photons are widely used as the information carrier in quantum communication. Various degrees of freedom of a photon can be used for qubit encoding. For example, the four BB84 states $\{\ket{0},\ket{1},\ket{+},\ket{-}\}$ can be encoded into four polarization states of a photon, namely, vertical, horizontal, $45^{\circ}$, and $135^{\circ}$, respectively.

\begin{mybox}[label={box:BB84}]{{BB84 protocol \cite{bennett1984proceedings}, a prepare-and-measure protocol.}}
	\begin{enumerate}[(1)]
		\item
		\emph{State preparation:} Alice randomly prepares qubits in one of the four states, $\ket{0}$, $\ket{1}$, $\ket{+}$, and $\ket{-}$, where $\ket{0}$ and $\ket{1}$ form the $Z$ basis and $\ket{\pm}=(\ket{0}\pm\ket{1})/\sqrt2$ form the $X$ basis.
		\item
		\emph{State transmission and measurement:} Alice sends the encoded qubits to Bob through a quantum channel. Bob measures each received qubit in the $Z$ or $X$ basis randomly.
		\item
		\emph{Key sifting:} Alice and Bob announce their choices of bases publicly through an authenticated classical channel. They keep only the bits where they use the same bases and discard the rest. They accumulate $n$ sifted key bits.
		\item
		\emph{Key distillation:} Alice and Bob perform classical postprocessing, including information reconciliation and privacy amplification, to generate a secret key from the $n$-bit sifted key.
	\end{enumerate}
\end{mybox}

When Alice prepares a qubit in the $Z$ basis and Bob measures in the same basis, an error occurs when their bit values are different. The errors in the $Z$ basis are called \emph{bit errors}. Similarly, when they operate in the $X$ basis, a \emph{phase error} occurs when their bits are different. Let us denote the bit- and phase-error rates by $e_b$ and $e_p$, respectively:
\begin{equation} \label{eq:biterror}
	\begin{aligned}
		e_b &= \frac{\text{number of bit errors}}{n}, \\
		e_p &= \frac{\text{number of phase errors}}{n}. \\
	\end{aligned}
\end{equation}
Measurements on the two bases are noncommuting. Due to quantum mechanics, with any attempt to extract nontrivial information from the qubits, Eve would inevitably introduce disturbance, such as errors, making $e_b,e_p\neq0$. Intuitively, Alice and Bob share a private key if $e_b=e_p=0$.

Before the security analysis, we introduce an entanglement-based protocol with EPR pairs, BBM92, in Box \ref{box:BBM92}. The BBM92 protocol can be reduced to the BB84 protocol if Alice measures her half of the state in the $Z$ or $X$ basis. Based on her measurement result, Alice equivalently sends a qubit in $\ket{0}$, $\ket{1}$, $\ket{+}$, or $\ket{-}$ to Bob.

\begin{mybox}[label={box:BBM92}]{{BBM92 protocol \cite{bennett1992quantum}, an entanglement version of BB84.}}
	\begin{enumerate}[(1)]
		\item
		\emph{State preparation: }Alice prepares EPR pairs, $\ket{\Phi^+}=(\ket{00}+\ket{11})/\sqrt2$. She stores one half of the pairs locally and sends the other half to Bob.
		\item
		\emph{State storage: }Upon receiving a qubit, Bob stores the state in quantum memories. If the qubit has been lost in the channel or the quantum storage fails, they discard the pair.
		\item
		\emph{Quantum error correction: }Alice and Bob measure a random sample of the stored qubit pairs to estimate the quantum bit- and phase-error rates, $e_b$ and $e_p$. They apply quantum error correction to the remaining stored qubit pairs. They share $n$ (almost) perfect EPR pairs.
		\item
		\emph{Key measurement: }Both Alice and Bob measure the EPR pairs in the local $Z$ basis to obtain the final key.
	\end{enumerate}
\end{mybox}


For the key-distribution task, Alice and Bob need to make sure that their key bit strings are \emph{identical} and \emph{uniformly random} from Eve's point of view.
To satisfy these two requirements, the ideal key state shared by Alice and Bob, and any outside adversary Eve, is defined to be the following classical-classical-quantum state:
\begin{equation} \label{eq:idealqkey}
	\rho_{ideal} = 2^{-n} \sum_{k\in\{0,1\}^n} \ket{k,k}_{A,B}\bra{k,k} \otimes \rho_E,
\end{equation}
where systems $A,B$ are keys held by Alice and Bob and system $E$ is held by Eve. In the ideal key state, Alice's and Bob's key bit strings are identical. Eve's system $\rho_E$ is independent of the key $k$, which brings her no more information about the key string than a random guess. This definition follows the works of Ben-Or \emph{et al.} \cite{ben2005universal} and Renner and K\"{o}nig \cite{Renner2005Security}.

In practice, however, Alice and Bob cannot generate an ideal key. It is reasonable to allow for a small failure probability. That is, Alice and Bob can generate a key state that is very close to an ideal one. To put the idea into a rigorous form, if the realistic key state
\begin{equation} \label{eq:actualkey}
	\rho_{key} = \sum_{k_A,k_B\in\{0,1\}^n}\Pr(k_A,k_B) \ket{k_A}_{A}\bra{k_A}\otimes\ket{k_B}_{B}\bra{k_B} \otimes \rho_E(k_A,k_B),
\end{equation}
satisfies
\begin{equation} \label{eq:tracedis}
	\dfrac{1}{2}	\min_{\rho_E}\|\rho_{key}-\rho_{ideal}\|_1\leq\varepsilon,
\end{equation}
the protocol will be $\varepsilon$ secure. Here, the distance measure is the trace distance. For two density matrices $\rho,\sigma$, the measure is defined as

\begin{equation} \label{eq:deftracedis}
	\begin{aligned}
		\dfrac{1}{2} \|\rho - \sigma\|_1 = \frac12\sum_i|\lambda_i|,
	\end{aligned}
\end{equation}
where $\lambda_i$ are eigenvalues of the operator $\rho-\sigma$.
The choice of the trace-distance measure is that it satisfies the requirements of a composable-security framework~\cite{ben2005universal,Renner2005Security}.

\begin{definition}[QKD $\varepsilon$ security]\label{def:securekey}
	A QKD protocol is $\varepsilon$ secure if the generated state $\rho_{key}$ given in Eq.~\eqref{eq:actualkey} is $\varepsilon$ close to the ideal key state $\rho_{ideal}$ given in Eq.~\eqref{eq:idealqkey} with respect to the trace-distance definition.
\end{definition}

This definition, usually referred to as \emph{soundness}, guarantees that once the protocol is not aborted, the generated key is private with a high probability. In the composable-security framework, there is another security parameter, \emph{completeness}--i.e., the protocol success probability--which guarantees that one protocol is not trivial and that it does not always abort. For most protocols, completeness can be easily established, so we do not discuss the completeness parameter in this work. To connect the security definition and the EPR pairs, we employ the following lemma.

\begin{lemma}[Lemma 1 in~\cite{Fung2010Finite}]\label{lemma:epsilonclose}
	If Alice and Bob share a quantum state $\rho_{AB}$ that is $\varepsilon_f$ close to the ideal key state before projective measurements onto $\ket{\Phi^+}^{\otimes n}$, i.e.
	\begin{equation} \label{eq:rhoABf}
		\begin{aligned}
			\bra{\Phi^+}^{\otimes n} \rho_{AB} \ket{\Phi^+}^{\otimes n} \ge 1-\varepsilon_f,
		\end{aligned}
	\end{equation}
	then the QKD protocol is $\varepsilon$ secure with $\varepsilon=\sqrt{\varepsilon_f (2 - \varepsilon_f)}$.
\end{lemma}

Combining Lemma~\ref{lemma:epsilonclose} and the security definition, we can conclude that if the state $\rho_{AB}$ shared by Alice and Bob before the measurement of a QKD protocol satisfies Eq.~\eqref{eq:rhoABf} with a small $\varepsilon_f$, this protocol is $\varepsilon$ secure. Therefore, the security of QKD is closely related to entanglement and the purpose of the security proof is to realize Eq.~\eqref{eq:rhoABf}.

\subsection{Security analysis based on quantum error correction}
In establishing the security analysis of QKD, important tools such as entanglement distillation \cite{Bennett96Mixed} and quantum error correction have been proposed and developed. A profound discovery is the connection between key privacy and entanglement. The number of generated keys can be elegantly linked with distillable entanglement under local operations and classical communication \cite{lo1999unconditional}. One can distill almost perfect entanglement via quantum bit- and phase-error correction before measuring the quantum state to obtain the final key. Though security can be proven via entanglement distillation, one does not need to carry out this procedure in reality. Quantum bit- and phase-error correction can be carried out using classical means once the parameters are well estimated in the virtual quantum scenario~\cite{shor2000simple}. In this framework, quantum bit-error correction corresponds to information reconciliation and quantum phase-error correction is transformed into privacy amplification. Below, we briefly introduce the security analysis based on quantum error correction \cite{lo1999unconditional,shor2000simple}.
In the following discussions, we call a state transmission and measurement that generates a pair of raw key bits a QKD \emph{round}. After many rounds, Alice and Bob accumulate enough raw key bits as a \emph{block} for postprocessing, which we call a QKD \emph{session}.

Alice and Bob aim to establish a perfect Einstein-Podolsky-Rosen (EPR) pair $\ket{\Phi^+}=(\ket{00}+\ket{11})/\sqrt2$ for each round, where $\ket{0}$ and $\ket{1}$ are the eigenstates of the Pauli operator $\sigma_z$. Due to channel disturbance or Eve's interference, the EPR pairs shared by Alice and Bob after $n$ rounds of state transmission are usually imperfect. We denote the state of these data pairs as $\rho_{AB}$. The difference between $\rho_{AB}$ and $\ketbra{\Phi^+}^{\otimes n}$ can be seen as disturbance and can be characterized by bit-error rate $e_b$ and the phase-error rate $e_p$:
\begin{equation} \label{eq:BPerr}
\begin{aligned}
e_b &= \frac12(1-\Tr[\rho_{AB}(\sigma_z\otimes\sigma_z)^{\otimes n}]),\\
e_p &= \frac12(1-\Tr[\rho_{AB}(\sigma_x\otimes\sigma_x)^{\otimes n}]),
\end{aligned}
\end{equation}
where $\sigma_x$ is another Pauli operation and we take the $Z$ basis for key generation. The values of $e_b$ and $e_p$ can be obtained by parameter estimation.
Here, the bit- and phase-error rates defined in Eq.~\eqref{eq:BPerr} are consistent with those in Eq.~\eqref{eq:biterror}. Note that Eq.~\eqref{eq:biterror} defines error frequencies, while Eq.~\eqref{eq:BPerr} defines error probabilities. Strictly speaking, these two definitions are only the same in the infinite-data-size limit, $n\rightarrow\infty$. In the finite-data-size regime, there is a deviation between them, caused by statistical fluctuations. For simplicity, we ignore the difference for the moment and we take it into account in the evaluation of the failure probability of quantum error correction.

If $e_b=e_p=0$, this means that for the pair of qubits $\rho$ in each round,
\begin{eqnarray}
	\Tr[\rho(\sigma_z\otimes\sigma_z)]&=&1,\label{eq:aimofQECbit} \\
	\Tr[\rho(\sigma_x\otimes\sigma_x)]&=&1,\label{eq:aimofQECphase}
\end{eqnarray}
which indicates that $\rho=\ketbra{\Phi^+}$. To realize Eqs.~\eqref{eq:aimofQECbit} and \eqref{eq:aimofQECphase}, Alice and Bob can apply the quantum circuit shown in Figure \ref{fig:QECprocedure} to do the quantum bit- and phase-error correction. Quantum bit-error correction guarantees Eq.~\eqref{eq:aimofQECbit} and quantum phase-error correction guarantees Eq.~\eqref{eq:aimofQECphase}.

\begin{figure}[hbt!]
\centering \includegraphics[width=16cm]{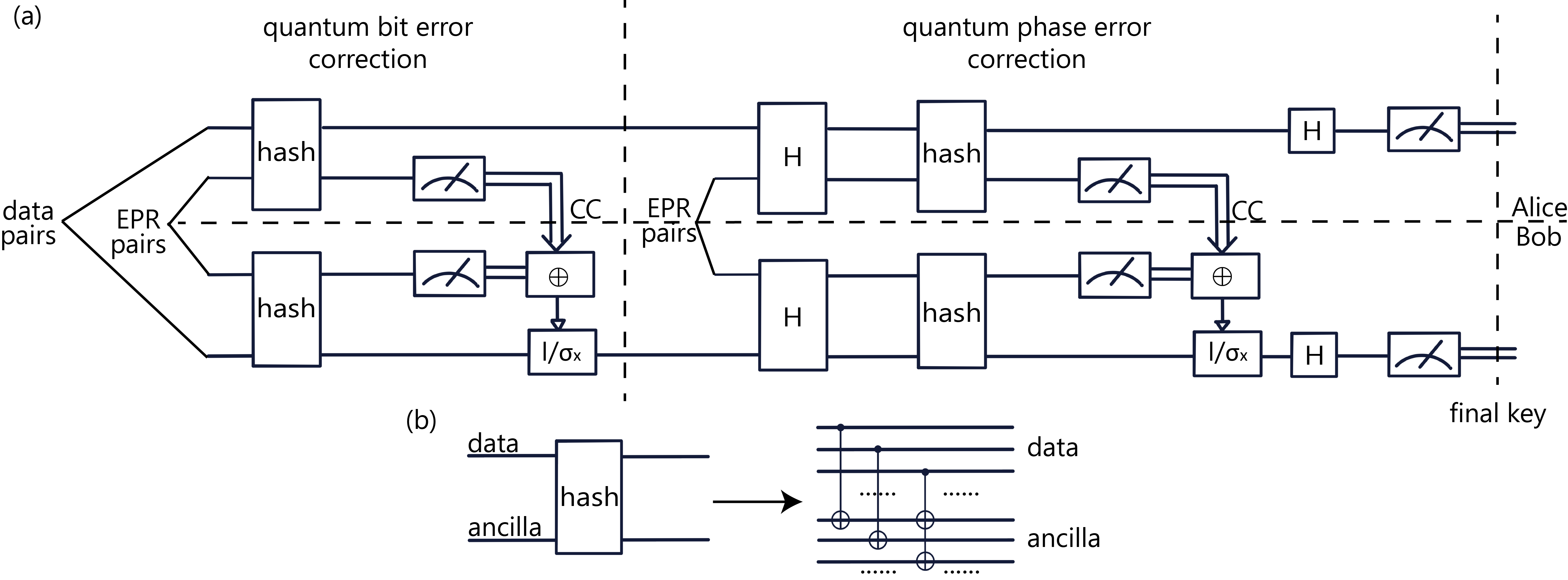}
\caption{(a) Quantum bit- and phase-error correction.
The measurements in all the figures are $Z$-basis measurements by default. ``CC" is short for classical communication. The $\oplus$ operation means XOR operations on classical bit strings. $H$ represents the Hadamard gate applied to each of the involved qubits. $I/\sigma_x$ represents an identity or $\sigma_x$ operation on the qubits, depending on the error syndrome. (b) In the linear case, the hash functions can be represented by matrices and realized by a series of controlled-NOT (CNOT) gates between the data (as control) and ancillary (as target) qubits. The measurement outcomes of ancillary qubits would give the parity information of the data qubits.}\label{fig:QECprocedure}
\end{figure}

The cost of quantum error correction comes from the ancillary EPR pairs used during the procedure. As an analog to the cost in classical error correction, the number of ancillary EPR pairs equals the number of parity check bits. Then, bit- and phase-error correction will cost $t_b = nh(e_b)$ and $t_p = nh(e_p)$ ancillary EPR pairs, respectively, where $h(x)=-x\log(x)-(1-x)\log(1-x)$ is the binary entropy function. Throughout this work, all logarithms are base 2. These costs can be derived from the cost in classical error correction, the details of which we leave to Appendix \ref{app:EC}. Therefore, the net generation rate of EPR pairs is given by \cite{shor2000simple},
\begin{equation} \label{eq:ShorPreskillR}
	\begin{aligned}
		r = 1-h(e_b)-h(e_p).
	\end{aligned}
\end{equation}

We can further reduce the procedure to a prepare-and-measure one by moving the final measurement to the front of quantum error correction. The reduction of quantum bit-error correction is straightforward, since all the operations can be done equivalently on classical data bits and it finally becomes information reconciliation. In contrast, the final measurement cannot be moved ahead of the quantum phase-error correction directly. Alice and Bob can replace the final measurements with a properly chosen joint measurement, which can be moved ahead of quantum phase-error correction \cite{shor2000simple}. Then the joint measurement becomes a classical operation called privacy amplification. Finally, quantum error correction is reduced to a ``measurement + postprocessing'' procedure. The reduction is in the spirit of the Shor-Preskill security proof \cite{shor2000simple}. We leave the details to Appendix \ref{app:QECwithhash} and \ref{app:reduction}.

Due to the joint measurement, each final key bit depends on the measurement results from all the $n$ data qubits. Hence, for privacy amplification, Alice and Bob need to wait for all the quantum states to be transmitted and measured in a QKD session. We call this \emph{block} privacy amplification.


\section{Stream privacy amplification}\label{Sc:newPA}
\subsection{Reduction of quantum error correction}
In the aforementioned reduction, Alice and Bob cannot move the final key measurement ahead of phase-error correction directly. The main obstacle is that the Hadamard gate does not commute with the dephasing operation caused by the final key measurement,
\begin{equation}\label{eq:depahsingmeasure}
	\Delta_{Z^{\otimes n}}[\rho_{A}]=\sum_{\vec{k}\in \{0,1\}^n} \ketbra{\vec{k}} \rho_A \ketbra{\vec{k}}, 
\end{equation}
where $\rho_{A}$ denotes the state that Alice holds and the dephasing operation is defined with respect to the key-measurement basis $Z^{\otimes n}$ with outcomes $\vec{k}$. The operation on Bob's side is similar. In the Shor-Preskill reduction, Alice and Bob essentially construct a joint $Z$-basis measurement that commutes with hash operations to circumvent this problem. As a result, this makes privacy amplification operate in blocks. Here, we rearrange the reduction of the phase-error-correcting gates and keep the individual $Z$-basis measurements in the quantum phase-error correction. Consequently, we can render stream privacy amplification.

The key idea of the new reduction is to cancel all the Hadamard gates in quantum phase-error correction, shown in Figure \ref{fig:QECprocedure}. The controlled-NOT (CNOT) gate in the circuit always appears in pairs on Alice's and Bob's sides. We focus on one pair of CNOT gates in the quantum phase-error-correction part, as depicted in Figure \ref{fig:QPECtonewPA}(a). First, noting that $H^2=I$, we add two consecutive Hadamard gates after each output qubit of the CNOT gate. Then, the four Hadamard gates before and after each CNOT gate exchange the roles of the control and target qubits, $H^{\otimes 2}C_{\alpha\beta}H^{\otimes 2}= C_{\beta\alpha}$, where $C_{\alpha\beta}$ denotes a CNOT gate with control qubit $\alpha$ and target qubit $\beta$ and $C_{\beta\alpha}$ is the other way around. For Bob's data qubit, the phase-error-correcting operator $I/\sigma_x$ becomes $I/\sigma_z$ since $\sigma_z=H\sigma_xH$. Hence, we prove that circuits (a) and (b) in Figure \ref{fig:QPECtonewPA} are equivalent. Since the new phase-error-correcting operator $I/\sigma_z$ does not affect the $Z$-basis measurement, this step can be skipped along with the error-syndrome measurements on the ancillary qubits. The remaining operations commute with the dephasing operation, $\Delta_{Z^{\otimes n}}$. Alice and Bob can add $Z$-basis measurements on ancillary qubits after the CNOT gates, since they are irrelevant at that point. Finally, they can move the final measurement ahead of quantum error correction, as shown in Figure \ref{fig:QPECtonewPA}(c).


\begin{figure}[hbt!]
\centering \includegraphics[width=14cm]{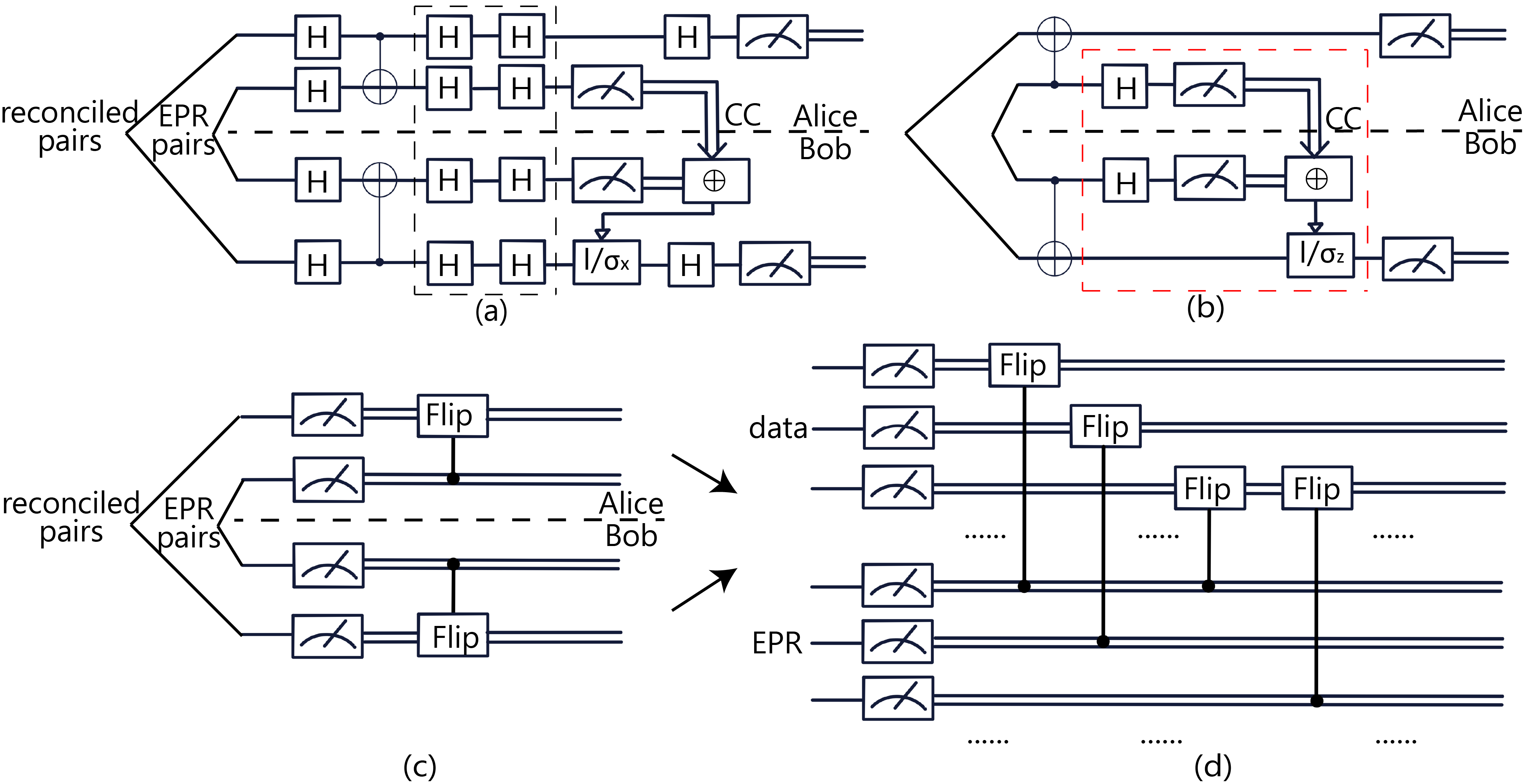}
\caption{Circuit (a) is derived from the quantum phase-error-correction phase in Figure \ref{fig:QECprocedure} by adding Hadamard gates in dashed boxes that form identity operations. We take one pair of CONT operations for illustration. Circuit (b) is equivalent to circuit (a) by considering the following facts: $H^{\otimes 2}C_{\alpha\beta}H^{\otimes 2}= C_{\beta\alpha}$, $H^2=I$, and $H\sigma_xH=\sigma_z$. Since neither the identity nor the $\sigma_z$ gate affects the $Z$-basis measurement, the operations in the dashed box of circuit (b) are redundant and can be removed. Then by moving the $Z$-basis measurement on ancillary qubits ahead of the hash operation and changing quantum-control flips to classical-control flips, circuit (b) turns into circuit (c), a ``measurement + postprocessing'' case. Circuit (d) shows the case of multiple CNOT pairs, taking the hashing circuit in Figure \ref{fig:QECprocedure} (b) as an example. In the end, both Alice and Bob employ circuit (d) to obtain final key strings.}
\label{fig:QPECtonewPA}
\end{figure}

So far, we have only considered one CNOT gate. The hash operation in phase-error correction shown in Figure \ref{fig:QECprocedure} is composed of many CNOT gates. This reduction also works for the general hash-operation case. With this argument, by inserting consecutive Hadamard gates $H^2=I$ after each CNOT gate of the phase-error-correction part in Figure \ref{fig:QECprocedure}, we can reduce the whole quantum error-correction circuit to the ``measurement + postprocessing'' case, as shown in Figure \ref{fig:QPECtonewPA}(d).

With the new reduction, the final key is determined by single-qubit measurements plus bit flips. The $Z$-basis measurement on the ancillary EPR pairs will provide Alice and Bob with a secure key \emph{seed}. The bit flips are controlled by the seed and the hashing matrix. Then, the $i$th final key bit, extracted from the $i$th data qubit, is independent of the other data qubits. Hence, the new procedure can output the final key in a stream, i.e., the users can obtain a secure key bit once a pair of raw key bits is reconciled successfully between Alice and Bob. Following the term ``stream cipher'' in classical cryptography, we call this \emph{stream} privacy amplification, as presented in Box \ref{box:streamPA}. The hashing matrix $M$ in step (1) is the transpose of the original hashing matrix used in the quantum phase-error-correction phase of Figure \ref{fig:QECprocedure}, because the original hashing matrix acts on $X$ basis while $M$ acts on the $Z$ basis. The cost of stream privacy amplification lies in step (2), where $nh(e_p)$ preshared secure bits are consumed. The final key rate matches Eq.~\eqref{eq:ShorPreskillR}. Note that we consider the infinite-data-size limit here. The finite data-size effects due to statistical fluctuations have been well considered in the literature and are shown in Appendix \ref{app:EC}.

\begin{mybox}[label={box:streamPA}]{{Stream privacy amplification}}
After information reconciliation, denote Alice's and Bob's reconciled key as $\vec{a}\in\{0,1\}^{n}$. 
\begin{enumerate}
\item
Alice and Bob randomly choose a hashing matrix $M$ of size $nh(e_p)\times n$. 
\item
Alice and Bob use an $nh(e_p)$-bit seed, $\vec{d}\in\{0,1\}^{nh(e_p)}$, to generate a pseudorandom string, $\vec{d} \cdot M$, where the dot product between the row vector and the matrix needs to take modulo 2 addition.
\item
The final key is given by $\vec{k}=\vec{d}\cdot M \oplus \vec{a}$.
\end{enumerate}
\end{mybox}

With the above deduction, the failure probability of privacy amplification $\varepsilon$ is given by that of quantum phase-error correction. Given a phase-error pattern from a typical set determined by parameter estimation, a randomly chosen hash function can identify it with a high probability $1-\varepsilon$. This is the reason why Alice and Bob need a random hashing matrix in step (1) of Box \ref{box:streamPA}. Details of error correction and its failure probability are presented in Appendix \ref{app:EC}.

Note that the phase-error pattern must be set before choosing the random hash function. That is, Eve cannot know the hashing matrix before Alice and Bob obtain the raw key bits. Naively, Alice can randomly pick up a matrix and send it to Bob via an authenticated but nonencrypted channel, as in the conventional case in block privacy amplification. Since this public transmission of the hashing matrix must be done after quantum measurement, Alice and Bob have to wait for the whole block to be transmitted and measured. Then, they lose the ``stream" property in privacy amplification.

To solve this problem, we apply a different approach, in which Alice and Bob generate an identical random hashing matrix locally with a preshared key and never reveal it in public. Then, they can prepare this matrix [step (1)] and the pseudorandom string [step (2)] before quantum transmission. A naive implementation of this approach, in which Alice and Bob generate $M$ and $\vec{d}$ in each run of the privacy amplification, could consume too many preshared secure bits, as for most of the universal hashing matrices, the number of random bits required to generate the matrix is larger than the data size $n$. Fortunately, with the following theorem, Alice and Bob can reuse the private hashing matrix in multiple QKD sessions with a failure probability that increases linearly, satisfying the composable-security definition~\cite{ben2005universal,Renner2005Security}. Since the failure probability can be exponentially small, the same hashing matrix can be used for many QKD sessions. Therefore, the cost of generating this hashing matrix is shared with these sessions, making the average cost negligible.

\begin{theorem}[Reuse of hashing matrix in privacy amplification]\label{thm:reuse}
Given a QKD session, the failure probability of a randomly chosen hashing matrix for privacy amplification is upper bounded by $\varepsilon$. Then, for $m$ QKD sessions, if Alice and Bob apply the same randomly chosen matrix for each session, the probability that privacy amplification fails in at least one session is upper bounded by $m\varepsilon$.
\end{theorem}

\begin{proof}
Following the aforementioned deduction of quantum error correction, the failure probability of privacy amplification is determined by phase-error correction. Now, the question becomes that given $m$ phase-error patterns, if Alice and Bob randomly pick a hash function to correct all the errors, what is the failure probability that at least one phase-error pattern is unsuccessfully corrected? In Appendix \ref{App:Reuse}, we provide the failure probability of reusing hash functions for error correction, as given by Lemma \ref{lemma:npatterns}. Using this result, the answer to the above question is $m\varepsilon$.
\end{proof}

Note that this conclusion is consistent with lemma III.1 in \cite{Ma2013Extractor}, which can obtain similar findings in randomness extraction. Theorem 1 in \cite{Renner2005Security} also implies the reuse of seeds in QKD schemes.

Before running QKD sessions, Alice and Bob can perform steps (1) and (2) in Box \ref{box:streamPA} and prepare the pseudorandom string in advance. They only need to run step (3) in privacy amplification during real-time QKD, which is essentially composed of simple XOR operations and much faster than hash operations. In block privacy amplification, the computational complexity of the matrix multiplication with Toeplitz hashing is $O(n\log n)$ with the fast-Fourier-transform algorithm, where $n$ is the length of the reconciled key string~\cite{ma2011explicit,hayashi2016more}. In contrast, the computational complexity of step (3) is $n$ and hence stream privacy amplification is faster in real-time QKD, especially when the data size is large. 

In reality, Alice and Bob need parameter estimation before privacy amplification. This might restrict the ``stream" feature since the parameters cannot be accurately estimated until the whole data block is accumulated. Nevertheless, when the link between Alice and Bob is stable, they can foresee the parameters and apply them to stream privacy amplification. In practice, it is not difficult to maintain stability in a quantum communication
network~\cite{chen2021integrated}. In addition, the users can double check the parameters after the transmission of the whole block. If the predicted parameters are within a reasonable range, they keep the key. Otherwise, if the actual parameters show that the implemented privacy amplification cannot guarantee security, this implies that the length of the seed chosen in step (2) of Box \ref{box:streamPA} is inefficient. Alice and Bob can then choose another seed, according to the difference between the predicted and actual parameters, to carry out additional privacy amplification on the key to make it secure.

Note that the stream scheme can work for any block size in QKD implementations. In order to make privacy amplification efficient, Alice and Bob can employ a large data size without causing delays in real-time key generation. 
Compared with the previous ones, the unique feature of the new scheme---stream output---can make QKD more practical in scenarios such as the satellite-to-ground link \cite{liao2018satellite}.

Moreover, the bit-error locations in the input string will remain the same after stream privacy amplification, since the final key bit is only determined by the pseudorandom bit and the raw key bit at the same location. As a result, the errors will not spread out, and then privacy amplification can even be performed ahead of information reconciliation. This feature increases the flexibility of data postprocessing. For example, privacy amplification and information reconciliation can be performed in parallel. The recently proposed scenario of distributed private randomness distillation \cite{yang2019distributed} is also a potential application of the new scheme.

\subsection{Application I: Enhancing the security of trusted-relay QKD network}
The major issue of a trusted-relay QKD network lies in the trustworthiness of the intermediate nodes. There are some attempts to reduce the requirement on trusted relays, one of which is delayed privacy amplification \cite{fung2012quantum,Stacey2015Relay}. In the normal case, all the QKD links between two end users, Alice and Bob, will generate secure keys between neighboring nodes. Then, the intermediate relays swap the keys by announcing the XOR results of two keys generated with two neighbors. In the delayed privacy amplification case, the relays swap the key right after information reconciliation. Then, Alice and Bob perform privacy amplification without the relays. In this case, they can eliminate the relays from the final key generation process. That is, the relays do not obtain the final key directly. Of course, if the relays listen to classical communication to obtain privacy amplification matrices, they can still obtain the whole final key string. Therefore, the delayed privacy amplification scheme only works for \emph{honest but curious} relays.

Now, we can combine stream privacy amplification with delayed privacy amplification to further reduce the trustworthiness of the intermediate relays. After information reconciliation, all relays swap their keys by announcing the XOR results of two neighboring keys. Then, Alice and Bob will share a reconciled key string $\vec{a}$, which is also known to the relays. Note that Alice and Bob perform the steps in Box \ref{box:streamPA} locally. In particular, in step (2), they generate the pseudorandom string $\vec{d} \cdot M$ privately. Hence, the relays cannot know the final key without $\vec{d} \cdot M$. If the relays want to learn the final key, they need to figure out $\vec{d}$ and $M$. The seed $\vec{d}$ is private and changes in every QKD session. The hashing matrix $M$, on the other hand, is reused for many sessions in stream privacy amplification, so the relays might figure it out from final and reconciled key strings in past sessions by methods such as differential cryptanalysis. These analysis methods often consume a lot of computational resources. For an even higher security level with fewer assumptions on the relays, we can add another layer of security based on the computational complexity on the intermediate nodes. In practice, this combined scheme further reduces the requirement of the trustworthiness of the relays and enhances the security of trusted-relay QKD networks.


\subsection{Application II: Information-theoretic toolbox for classical encryption analysis}\label{subsec:newtool}
There is an interesting property of the pseudorandom string $\vec{d} \cdot M$ generated in step (2) of Box \ref{box:streamPA}. On the one hand, parameter estimation provides an upper bound on the information leakage of $nh(e_p)$ bits about the raw key string $\vec{a}$. On the other hand, the security proof guarantees that the information leakage is removed via the simple XOR operation $\vec{k}=\vec{d}\cdot M \oplus \vec{a}$. This implies that $\vec{d} \cdot M$ has at least an $nh(e_p)$-bit uncertainty to Eve. In the security analysis, we do not assume in advance which part of the $nh(e_p)$-bit information on $\vec{a}$ is known by Eve. Then, for any $nh(e_p)$ bits from $\vec{d} \cdot M$, the corresponding $nh(e_p)$-bit substring will be uniformly distributed from Eve's point of view, as shown in Figure \ref{fig:kwise}. This property is called $k$-wise independence in classical cryptography and is an essential requirement of many stream ciphers~\cite{alon1992simple}.

\begin{figure}[hbt!]
\centering \includegraphics[width=9cm]{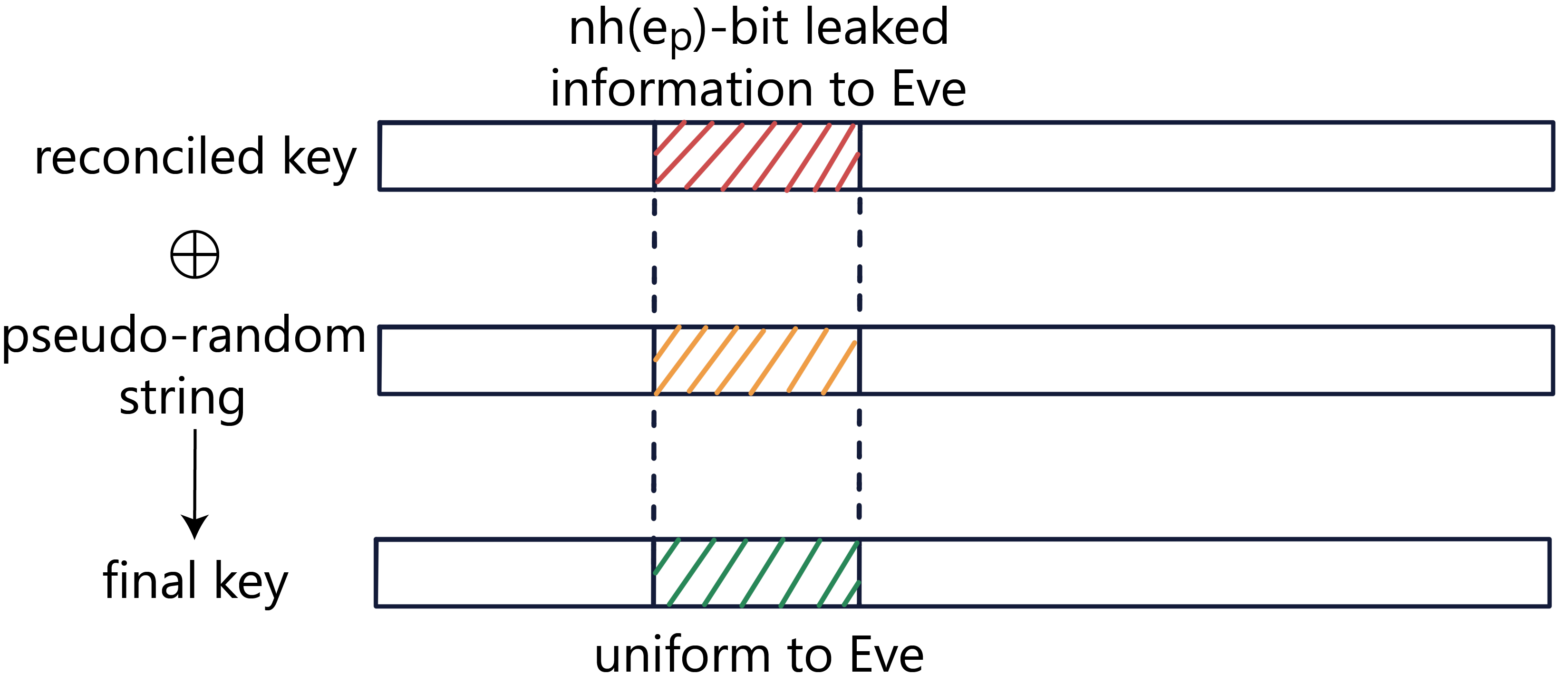}
\caption{An illustration of the pseudorandomness property of $\vec{d} \cdot M$. The red-shaded area in the reconciled key denotes the $nh(e_p)$-bit information leakage to Eve. This leaked information can be the bit values of the key or the parity information about the key. After the XOR of the reconciled key and the pseudorandom string $\vec{d} \cdot M$, the reconciled key becomes the final secure key and is uniformly distributed from Eve's point of view. That is, Eve's knowledge about the key is removed. As a corollary, the corresponding yellow-shaded area in the pseudorandom string should be uniform to Eve as well.}
\label{fig:kwise}
\end{figure}

The above observation inspires a new tool for the security analysis of classical encryption based on quantum phase-error correction. First, we note that the hash function for phase-error correction is not necessarily linear. Alice and Bob can employ an arbitrarily nonlinear code to extend the seed to a pseudorandom string in step (2) of Box \ref{box:streamPA}, as shown in Figure \ref{fig:randomvsEC}(a). This extension operation can be treated as a pseudorandom number generator. Second, we can further generalize this stream privacy amplification as a joint function of the reconciled key and the seed, as shown in Figure \ref{fig:randomvsEC}(b). The joint operation can be treated as a classical encryption box. Alice and Bob can employ sophisticated schemes here, such as the advanced encryption standard (AES) and lattice-based encryption algorithms. Third, we can express the classical operation in the quantum form, as a joint operation $\Lambda$ acting on data and ancillary qubits, as shown in Figure \ref{fig:randomvsEC}(c). Since the operation is classical, $\Lambda$ commutes with the dephasing operation in Eq.~\eqref{eq:depahsingmeasure}: $\Lambda[\Delta_{Z^{\otimes n}}(\rho)]=\Delta_{Z^{\otimes n}}[\Lambda(\rho)]$. By definition, $\Lambda$ is a dephasing incoherent operation (DIO) \cite{streltsov2017colloquium}. At last, we can move $\Lambda$ ahead of measurement. Since phase-error correction is virtual in the security analysis, Alice and Bob can add an extra operation $\Omega$ on the ancillary qubits if necessary, as shown in Figure \ref{fig:randomvsEC}(d).

\begin{figure}[hbt!]
\centering \includegraphics[width=12cm]{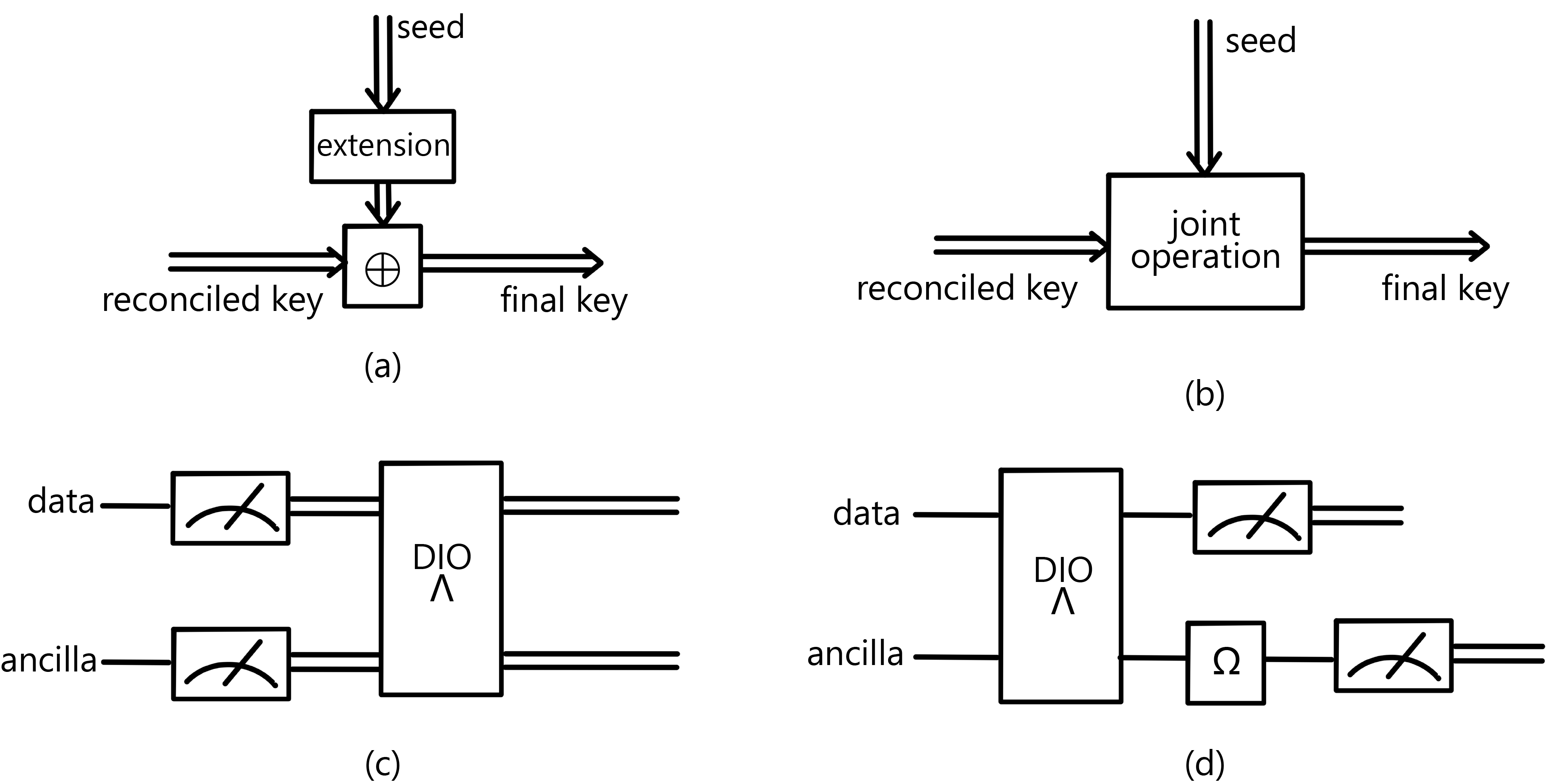}
\caption{(a) A schematic diagram for stream privacy amplification. The process of extending the seed can be treated as a pseudorandom number generator in step (2) of Box \ref{box:streamPA}. (b) The generalization of the extension and XOR operations as a joint operation, which can be seen as a classical encryption box. (c) The quantum form of (b). The classical joint operation becomes a DIO $\Lambda$. (d) Changing the order of $\Lambda$ and measurement. The final measurement on ancillary qubits can be arbitrary with an extra operation $\Omega$.}
\label{fig:randomvsEC}
\end{figure}

Similar to stream privacy amplification, we can analyze the security of the final output of Figure \ref{fig:randomvsEC}(d) by going back to the Lo-Chau security proof based on quantum error correction. The DIO $\Lambda$ is determined by the classical encryption algorithm, which is viewed as the encoding for a quantum error-correcting code. With this code, Alice and Bob exchange the measurement results of ancillary qubits as a phase-error syndrome which should give information about the phase errors of data qubits. According to the security analysis, Alice and Bob do not need to actually correct phase errors since this will not affect the final key. Then, Alice and Bob can add an extra virtual operation $\Omega$ on the ancillary qubits between $\Lambda$ and measurement. They can optimize  $\Omega$ to maximize the error-correction capability led by $\Lambda$. The security of the final output of the classical encryption will be determined by the phase-error-correction capability of the corresponding code. Since we consider the most general attacks in QKD, the method provides a generic information-theoretic analysis tool for classical encryption algorithms.

\subsection{Application III: Stream randomness extraction}
Besides QKD, privacy amplification also plays a vital role in many other quantum cryptographic tasks, such as QRNG. In general, the raw data generated from a practical QRNG system are not uniformly random. Due to device imperfection, some information about the raw data might even be leaked and might lead to potential security loopholes. Thus, user Alice needs to apply randomness extraction to the raw data. By definition, randomness extraction is essentially the same as privacy amplification. As a result, the stream privacy-amplification technique can also be directly applied to QRNG---stream randomness extraction. 

In QRNG, the amount of intrinsic randomness in the raw data is usually quantified in terms of min-entropy~\cite{ma2016quantum}. This randomness measure can be converted to the number of phase errors in a virtual quantum error-correction protocol~\cite{Tsurumaru2022equivalence}. Then, one can convert common block randomness extraction to a stream manner, as presented in Box~\ref{box:streamRE}.

\begin{mybox}[label={box:streamRE}]{{Stream randomness extraction}}
Alice generates a raw bit string from the QRNG device, denoted as $\vec{a}\in\{0,1\}^{n}$, the min-entropy of which is $H_{\min}$. 
\begin{enumerate}
\item
Alice randomly chooses a hashing matrix $M$ of size $(n-H_{\min})\times n$. 
\item
Alice uses an $(n-H_{\min})$-bit seed, $\vec{d}\in\{0,1\}^{n-H_{\min}}$, to generate a pseudorandom string, $\vec{d} \cdot M$, where the dot product between the row vector and the matrix needs to take modulo 2 addition.
\item
The final random bit string is given by $\vec{k}=\vec{d}\cdot M \oplus \vec{a}$.
\end{enumerate}
\end{mybox}	

For the data postprocessing of practical QRNG, stream randomness extraction can be a favored choice. In this context, Alice can fully characterize the quantum devices in use and hence has good empirical knowledge of the randomness generation rate. In other words, the min-entropy of output randomness can be well predicted in advance. With this property, steps (1) and (2) in Box~\ref{box:streamRE} can be done separately in advance without access to the raw data. This enables the main steps of postprocessing to be carried out in parallel with the generation of raw data, which can reduce the storage requirements and modularize quantum random number generators. In addition, the computational complexity of real-time postprocessing [step (3)] is only $n$ and hence helps to solve the current bottleneck of real-time random number generation---slow extraction.

\section{Discussion and Conclusions}\label{Sc:conclusion}
In this work, we propose a stream privacy-amplification scheme, where Alice and Bob locally generate a pseudorandom bit string and XOR it with the reconciled key to obtain the final key. This scheme has a stream output feature and hence can prevent unpleasant delay and error spreading in practice. In contrast to conventional block schemes, stream privacy amplification can be carried out ahead of information reconciliation, which makes the data postprocess more flexible. In addition, stream privacy amplification can enhance the security of a trusted-relay QKD network and improve the practicality of randomness extraction for quantum number generators.

We need to emphasize that although we reduce the stream privacy amplification from the Lo-Chau security proof, the technique is independent of security proofs. Other security-proof methods, such as Koashi's complementarity approach \cite{koashi2009simple}, can also be easily extended to the stream privacy-amplification case. Moreover, the concept is rather generic and can be applied to other QKD schemes, such as six-state, continuous-variable, measurement-device-independent, two-way communication postprocessing, and decoy-state \cite{xu2020secure} schemes. The practical issues, including realistic circumstances, hardware imperfections, and statistic fluctuations, will affect the parameter settings of stream privacy amplification, especially the length of the seed string and the size of the hashing matrix. One can combine with existing analysis methods to deal with these practical issues. In Appendix \ref{app:examplewithGLLP}, we give an example of employing stream privacy amplification in the Gottesman-Lo-L\"utkenhaus-Preskill (GLLP) framework~\cite{gottesman2004security}. The further applications of stream privacy amplification in other quantum cryptographic tasks like quantum oblivious transfer \cite{bennett1992practical,wei2017generic} are also worth studying.

Here, our proof is mainly based on phase-error correction. According to Ref.~\cite{Tsurumaru2022equivalence}, in general, this approach is equivalent to the one based on the quantum leftover hashing lemma~\cite{renner2008security}. An interesting direction is to reconsider the new scheme from the entropic approach point of view.

Our security analysis provides a new perspective to examine classical encryption algorithms information theoretically through quantum-information theories. Rigorous assessment of classical encryption algorithms, such as AES and lattice-based encryption, is often a formidable challenge. To the best of our knowledge, there has been little consideration to date in the context of the information-theoretic study of these encryption algorithms.

\section*{Acknowledgments}
We thank Guang Yang, Guoding Liu, Pei Zeng, and Hongyi Zhou for the helpful discussions.
This work was supported by the National Natural Science Foundation China under Grants No.~11875173 and No.~12174216 and by the National Key Research and Development Program of China under Grants No.~2019QY0702 and No.~2017YFA0303903.

\appendix

\section{Error correction}\label{app:EC}
We will briefly review error correction and derive its cost in this section.
Suppose Alice and Bob have two $n$-bit strings $\vec{x},\vec{y}\in \{0,1\}^{n}$ before error correction. The differences between these two strings are the errors to be corrected, represented by an error string $\vec{e}=\vec{x}\oplus\vec{y}$. Bob aims to reconcile his bit string to Alice's. For this purpose, the essential task is to locate all the errors, i.e., to figure out the error string $\vec{e}$. Denote each bit in an error string as a binary random variable, $E_i\in\{0,1\}$. An error in the $i$th bit is represented by $E_i=1$. The error string is the joint random variable of $\{E_i\}$, denoted as $\vec{E}$. We use the corresponding lowercase letters to denote the specific realizations of random variables. All the possible error strings form a linear space $\Omega$ with a size of $|\Omega|=2^n$. Denote the probability that an error string $\vec{e}\in \Omega$ occurs to be $p(\vec{e})$. In general, the joint random variable is not independent and identically distributed (i.i.d.).

Before error correction, Alice and Bob need to estimate the set of possible errors. In parameter estimation, we allow a failure probability $\varepsilon\geq0$. With the probability of $1-\varepsilon$, the error string will fall into the $\varepsilon$-smallest probable set $\mathcal{T}^\varepsilon_{\vec{E}}$. The aim of parameter estimation is to upper-bound the cardinality of the probable error set, or error cardinality, via the statistics of all measurements in an experiment \cite{Fung2010Finite}.

Given the error cardinality $|\mathcal{T}^\varepsilon_{\vec{E}}|$, Alice and Bob need to exchange a certain amount of parity information to correct the errors. There exists an error correction protocol to remove all errors by exchanging an amount of parity-check bits upper-bounded by
\begin{equation}\label{Eq:ErrorCorrectionCost}
	I_{ec}=\log|\mathcal{T}^\varepsilon_{\vec{E}}|-\log{\varepsilon_{ec}},
\end{equation}
where $\varepsilon_{ec}\geq0$ is the failure probability of error correction.

In the rest of this section, we will show how to derive Eq.~\ref{Eq:ErrorCorrectionCost} and calculate the cost $I_{ec}$ with the error rate. Let us start with the definition of the $\varepsilon$-smallest probable set.

\begin{definition}[$\varepsilon$-Smallest Probable Set] Given $0\leq\varepsilon<\dfrac{1}{2}$, a random variable $X\in \mathcal{X}$ which has a finite range $|\mathcal{X}|<\infty$ and probability mass function $p_X$, the $\varepsilon$-smallest probable set of $X$ is defined as the smallest set, $\mathcal{T}^\varepsilon_X$, such that a realization of $X$ lies in it with a failure probability no larger than $\varepsilon$,
	\begin{equation}
		\begin{aligned}
			\mathcal{T}^\varepsilon_X&=\arg \min_\mathcal{S}|\mathcal{S}|\\
			\text{such that}\; &\Pr(X\in \mathcal{S})\geq 1-\varepsilon.
		\end{aligned}
	\end{equation}
\end{definition}

To tackle the cardinality of $\varepsilon$-smallest probable set, one method is to use the weight of error strings. For an $n$-bit string $\vec{e}\in \{0,1\}^n$, we define a weight function, quantifying the number of $1$s in the string,
\begin{equation}
	wt(\vec{e}) = \sum_{i=1}^n e_i.
\end{equation}
We can bound the cardinality of a set via bounding the weights of the strings, as shown in the following lemma.

\begin{lemma}\label{Thm:SetCardinal}
	Given constants $c\geq0$, $0<r<1$, and $n\in\mathbb{Z}^+$, the cardinality of the $n$-bit string set,
	\begin{align}
		\mathcal{D}(c)&=\left\{\vec{e}\in\{0,1\}^n \mid wt(\vec{e})\in \mathcal{C}\right\}, \label{Eq:TypicalSetA} \\
		\mathcal{C}&=
		\begin{cases}	
			[0,nr+c], & r\le\frac{1}{2} \\
			[nr-c,n], & r>\frac{1}{2}
		\end{cases}, \label{eq:2intervals}
	\end{align}
	can be upper bounded by
	\begin{equation}\label{Eq:SetCardBound}
		|\mathcal{D}(c)| < 2^{nh(r)+c\left|\log\frac{r}{1-r}\right|}.
	\end{equation}
\end{lemma}

\begin{proof}
	When $r=1/2$, the bound in Eq.~\eqref{Eq:SetCardBound} is trivial. The case for $r>\dfrac{1}{2}$ is the same as the one for $r<\dfrac{1}{2}$ by switching the definitions of bits 0 and 1. So we only need to prove the lemma for $r<\dfrac{1}{2}$.
	
	We introduce a function of a binary variable $e\in\{0,1\}$,
	\begin{equation}\label{Eq:iidProb}
		p(e)=
		\begin{cases}	
			r, & e=1 \\
			1-r, & e=0
		\end{cases},
	\end{equation}
	and we can show that,
	\begin{equation} \label{eq:1SampleEnt}
		\sum_{\vec{e}\in\{0,1\}^n}2^{\sum_{i=1}^n\log{p(e_i)}} = \prod_{i=1}^{n}[p(e_i=0)+p(e_i=1)] = 1.
	\end{equation}
	For a bit string $\vec{e}=(e_1,\cdots,e_n)$ with a weight of $wt(\vec{e})=k$, we define a function,
	\begin{equation}\label{Eq:LogProb}
		g(k)=k\log{r}+(n-k)\log(1-r)=\sum_{i=1}^n\log{[p(e_i)]}.
	\end{equation}
	For $\vec{e}\in\mathcal{D}(c)$, we have $ k\in [0,nr+c]$ and $g(k)$ is a decreasing function of $k$ when $r<\dfrac{1}{2}$, so,
	\begin{equation}\label{Eq:StrSampleEnt}
		\begin{split}
			g(k)&\geq g(nr+c)\\
			&=(nr+c)\log{r}+(n-nr-c)\log(1-r) \\
			&=-nh(r)-c\left|\log{\dfrac{r}{1-r}}\right|.
		\end{split}
	\end{equation}
	Then, combining Eqs.~\eqref{eq:1SampleEnt}--\eqref{Eq:StrSampleEnt}, we have
	\begin{equation}
		\begin{aligned}
			1 &=\sum_{\vec{e}\in\{0,1\}^n}2^{\sum_{i=1}^n\log{p(e_i)}} \\
			&=\sum_{\vec{e}\in\{0,1\}^n}2^{g(k)}\\
			&\geq \sum_{\vec{e}\in\mathcal{D}(c)}2^{g(k)} \\
			&\geq 2^{-nh(r)-c\left|\log{\frac{r}{1-r}}\right|} |\mathcal{D}(c)|,
		\end{aligned}
	\end{equation}
	where the third and fourth inequalities cannot take equal signs simultaneously. Finally, we obtain Eq.~\eqref{Eq:SetCardBound}.
	
\end{proof}

Because the weight of an $n$-bit string $\vec{e}$ satisfies $0 \le wt(\vec{e})\le n$, one can combine the two intervals of Eq.~\eqref{eq:2intervals} in the lemma to get a smaller set, as given in the following corollary.

\begin{corollary}\label{Thm:SetCardinalSym}
	Given constants $c\geq0$, $0<r<1$, and $n\in\mathbb{Z}^+$, the cardinality of the $n$-bit string set,
	\begin{equation}
		\mathcal{D}(c)=\left\{\vec{e}\in\{0,1\}^n \mid wt(\vec{e})\in \left[nr-c,nr+c\right]\right\},
	\end{equation}
	can be upper bounded by
	\begin{equation}
		\begin{aligned}
			|\mathcal{D}(c)| < 2^{nh(r)+c\left|\log\frac{r}{1-r}\right|}.
		\end{aligned}
	\end{equation}
\end{corollary}

We remark that the result holds even when the random variables, $E_i$, associated with the set $\mathcal{D}(c)$ are arbitrarily correlated, not necessarily i.i.d.. Alice and Bob can use Lemma \ref{Thm:SetCardinal} to bound the cardinality of $\varepsilon$-smallest probable set if they can estimate $r$ and $c$. Before error correction, they can estimate the error rate with a failure probability of $\varepsilon$, say, via the random sampling method. Suppose they obtain an error frequency of $r$ in the test samples. Without loss of generality, we assume $r\le1/2$. In the asymptotic case $n\rightarrow\infty$, the number of errors in the data is given by $nr$. With a finite data size, the rate fluctuates around $r$. Alice and Bob can bound the number of errors, $wt(\vec{e})\le nr+c$, via the random sampling method, where $c/n$ represents the deviation of the error rate from the test samples. The deviation $c/n$ is usually related to the failure probability of parameter estimation $\varepsilon$ and typically has an order of $1/\sqrt{n}$ \cite{Fung2010Finite}. With $r$, $c$, and $n$, they can apply Lemma~\ref{Thm:SetCardinal} to determine the error cardinality.

Here, we introduce another method to upper bound the cardinality of typical error sets, which is tighter under more restricted conditions \cite{Fung2010Finite}.

\begin{lemma}[\cite{Fung2010Finite}]\label{lem:Fung13bd}
	Given constants $c, r\geq0$ and $n\in\mathbb{Z}^+$ satisfying $r+c/n\le\frac{1}{3}$, the cardinality of the $n$-bit string set,
	\begin{equation}
		\mathcal{D}(c)=\left\{\vec{e}\in\{0,1\}^n \mid wt(\vec{e})\in [0,nr+c)\right\},
	\end{equation}
	can be upper bounded by
	\begin{equation}\label{fredresult}
		|\mathcal{D}(c)|<2^{n\cdot h(r+\frac{c}{n})}.
	\end{equation}
\end{lemma}
\begin{proof}
	By definition, we have
	\begin{equation} \label{eq:SumBinoBd}
		\begin{aligned}
			|\mathcal{D}(c)| &= \sum_{0\le k<nr+c} \binom{n}{k}\\
			&< \binom{n}{\lceil nr+c\rceil}\\
			&\le 2^{n\cdot h(r+\frac{c}{n})}.
		\end{aligned}
	\end{equation}
	The first inequality holds when $r\leq\dfrac{1}{3}$.
\end{proof}

This lemma provides a tight bound because in Eq.~\eqref{eq:SumBinoBd},
\begin{equation} \label{eq:SumBinoBdul}
	\binom{n}{\lceil nr+c\rceil-1} \le \sum_{0\le k<nr+c} \binom{n}{k}<\binom{n}{\lceil nr+c\rceil},
\end{equation}
and the last inequality in Eq.~\eqref{eq:SumBinoBd} is also tight in the logarithm sense.

When the error rate deviation is small, $c/n\ll1$, we can take the Taylor series expansion of $h(r+\frac{c}{n})$ at $r$, we can get
\begin{equation}\label{Eq:taylor}
	\begin{split}
		h\left(r+\frac{c}{n}\right)&=h(r) +h'(r)\dfrac{c}{n} +\dfrac{h''(r)}{2!}\left(\dfrac{c}{n}\right)^2 +\cdots \\
		&=h(r)+\frac{c}{n}\log\frac{1-r}{r}+\dfrac{h''(r)}{2!}\left(\dfrac{c}{n}\right)^2+\cdots
	\end{split}
\end{equation}
By comparing Eqs.~\eqref{Eq:SetCardBound}, \eqref{fredresult}, and \eqref{Eq:taylor}, one can see that Lemma \ref{Thm:SetCardinal} is a first-order approximation of Lemma \ref{lem:Fung13bd}. Since $h''(r)<0$, Lemma \ref{lem:Fung13bd} provides a tighter bound than Lemma \ref{Thm:SetCardinal}. On the other hand, Lemma \ref{Thm:SetCardinal} does not require $r\le\frac{1}{3}$, so it can be applied to more general cases.

Now, we show how to locate the errors given the probable error set $\mathcal{T}^\varepsilon$, where $\varepsilon$ is the failure probability for parameter estimation, $\sum_{\vec{e}\in\mathcal{T}^\varepsilon}p(\vec{e})\geq 1-\varepsilon$. In the following, we shall introduce error correction based on universal hash functions.

\begin{definition}[Universal Hash family]
	A family of functions $\mathcal{F}$ mapping elements $\vec{e}$ in a space $\mathcal{T}$ to another space $\mathcal{S}$ is called a universal hash family if the probability of a randomly chosen hash function outputting the same hashing result for any two different strings is upper bounded by
	\begin{equation}\label{eq:universal2}
		\forall \vec{e}_i\neq \vec{e}_{j} \in \mathcal{T}, \;	\Pr_{f\in\mathcal{F}}\left[ f(\vec{e}_i)=f(\vec{e}_{j})\right]\leq \dfrac{1}{|\mathcal{S}|}.
	\end{equation}
\end{definition}

Given the probable error set $\mathcal{T}^{\varepsilon}$, Alice and Bob decide a universal hash family $\mathcal{F}$ and choose a hash function $f\in\mathcal{F}$ randomly. Alice applies the hash function to her bit string $\vec{x}$ and sends the hashing result $\vec{s}_A=f(\vec{x})$ to Bob, who calculates the syndrome $\vec{s}=f(\vec{e})$ based on $\vec{s}_A$, $\vec{y}$ and $f$. If the hash functions are linear, then
\begin{equation}
	\begin{aligned}
		\vec{s}&=\vec{s}_A\oplus\vec{s}_B \\
		&=f(\vec{x})\oplus f(\vec{y}) \\
		&=f(\vec{x}\oplus\vec{y}) \\
		&=f(\vec{e}).
	\end{aligned}
\end{equation}
The linear hash function is normally implemented by multiplication of a hashing matrix.

If there is only one error string $\vec{e}\in\mathcal{T}^{\varepsilon}$ that satisfies $f(\vec{e})=\vec{s}$, Bob can figure out $\vec{e}$ correctly in principle. The probability of Bob figuring out a wrong error string is upper-bounded by
\begin{equation}\label{eq:epsionEC}
	\begin{aligned}
		\sum_{\vec{e}'\in \mathcal{T}^{\varepsilon}\setminus\{\vec{e}\}}\Pr_{f\in\mathcal{F}}\left[f(\vec{e}')=f(\vec{e})\right]&\leq \sum_{\vec{e}'\in \mathcal{T}^{\varepsilon}\setminus\{\vec{e}\}}\dfrac{1}{|\mathcal{S}|}\\
		&<\dfrac{|\mathcal{T^\varepsilon}|}{|\mathcal{S}|} \\ &= 2^{-(I_{ec}-k)} \\
		&\equiv \varepsilon_{ec},
	\end{aligned}
\end{equation}
where $\mathcal{T}^{\varepsilon}\setminus\{\vec{e}\}$ denotes the subset of $\mathcal{T}^{\varepsilon}$ excluding the element $\vec{e}$, the first inequality comes from Eq.~\eqref{eq:universal2}, the second inequality comes from $|\mathcal{T}^{\varepsilon}\setminus\{\vec{e}\}|=|\mathcal{T^\varepsilon}|-1$, $I_{ec}=\log|\mathcal{S}|$ is the effective length of syndrome $\vec{s}$, and $k=\log|\mathcal{T^\varepsilon}|$ is the logarithm of error cardinality. The failure probability $\varepsilon_{ec}$ approaches to zero exponentially with respect to $I_{ec}-k$.
If $I_{ec}-k$ is large enough, the failure probability will become negligible such that Bob can figure out $\vec{e}$ almost surely and then perform error correction according to the syndrome. 
We emphasize that $\varepsilon_{ec}$ is the failure probability for error correction and is different from the failure probability $\varepsilon$ for parameter estimation.
We can  derive Eq.~\eqref{Eq:ErrorCorrectionCost} from Eq.~\eqref{eq:epsionEC} by taking the logarithm,
\begin{equation}
	-I_{ec}-\log|\mathcal{T^\varepsilon}|=\log\varepsilon_{ec}.
\end{equation}
The cardinality of $T^\varepsilon$ is given by Lemma \ref{Thm:SetCardinal}. Then, we have
\begin{equation}
	\begin{split}
		\dfrac{1}{n}I_{ec}&=\dfrac{1}{n}\log|\mathcal{T^\varepsilon}|-\dfrac{1}{n}\log\varepsilon_{ec}\\
		&<h(r)+\dfrac{c}{n}\left|\log\frac{r}{1-r}\right|-\dfrac{1}{n}\log\varepsilon_{ec}\\
		&\overset{n\rightarrow\infty}{\longrightarrow} h(r).
	\end{split}
\end{equation}
In the third line, we take the asymptotic limit, $\dfrac{c}{n}\rightarrow 0$ because normally in parameter estimation $c$ has an order of $1/\sqrt{n}$, and $\dfrac{1}{n}\log\varepsilon_{ec}\rightarrow 0$ since $\varepsilon_{ec}$ is a constant. This gives Shannon's source encoding theorem in classical information theory.

\section{Reuse of hash function} \label{App:Reuse}
In the derivation of Eq.~\eqref{eq:epsionEC}, we implicitly assume errors cannot depend on the choice of the hash function. This requirement can be further understood in an adversary scenario, where the error string $\vec{e}$ is determined by an adversary Eve. In this case, Eve does not know the hash function chosen by Alice and Bob before she fixes the error string $\vec{e}$, or, before Alice and Bob obtain their bit strings $\vec{x}$ and $\vec{y}$, respectively. Otherwise, with prior knowledge of the hash function, Eve can choose the error string craftily such that it lies outside the error space that the chosen hash function can handle. As a result, Alice and Bob cannot figure out the correct error string $\vec{e}$, and the error correction would fail. To guarantee this requirement, there are two possible ways for Alice and Bob to decide the hash function.
\begin{enumerate}
	\item
	Alice and Bob decide the hash function after the error string is fixed, i.e., after the bit strings $\vec{x},\vec{y}$ are obtained. Then, they exchange some random bits through a public channel to determine a hash function randomly.
	\item
	Alice and Bob decide the hash function randomly by consuming some preshared private randomness and keep the hash function secret from Eve. In later error correction, they can reuse the same matrix at the cost that the total failure probability increases linearly with the number of sessions. We state the result formally in Lemma \ref{lemma:npatterns}.
\end{enumerate}

\begin{lemma}\label{lemma:npatterns}
	Given a set of error strings $\mathcal{T}^{\varepsilon}$ and a family of hash functions $\mathcal{F}$, $\forall\vec{e}\in\mathcal{T}^{\varepsilon}$, suppose the failure probability of a randomly chosen hash function to identify $\vec{e}$ is upper bounded by $\varepsilon_{ec}$. For any $m$ error strings in $\mathcal{T}^{\varepsilon}$, the failure probability of a randomly chosen hash function to simultaneously identify all the $m$ error strings in each session is upper bounded by $m\varepsilon_{ec}$.
\end{lemma}

\begin{proof}
	$\forall \vec{e}\in\mathcal{T}^{\varepsilon}$, the failure probability of error correction is given by Eq.~\eqref{eq:epsionEC},
	\begin{equation}
		\Pr_{f\in \mathcal{F}}\left[\exists\vec{e}'\in\mathcal{T}^{\varepsilon}\setminus\{\vec{e}\}, s.t., f(\vec{e}')=f(\vec{e})\right] \leq \varepsilon_{ec},
	\end{equation}
	where the probability is defined in the hashing family. Then, the failure probability of identifying $m$ error strings simultaneously is given by,
	\begin{equation}
		\begin{aligned}
			&\Pr_{f\in \mathcal{F}}\left[\exists\vec{e}'\in\mathcal{T}^{\varepsilon}, s.t., \vec{e}'\neq\vec{e}_1, f(\vec{e}')=f(\vec{e}_1) \text{ or } \vec{e}'\neq\vec{e}_2, f(\vec{e}')=f(\vec{e}_2) \text{ or } \dots \text{ or }  \vec{e}'\neq\vec{e}_m, f(\vec{e}')=f(\vec{e}_m) \right] \\
			&\leq \sum_{i=1}^{m} \Pr_{f\in \mathcal{F}} \left[\exists\vec{e}'\in\mathcal{T}^{\varepsilon}\setminus\{\vec{e}_i\}, s.t., f(\vec{e}')=f(\vec{e}_i)\right] \\
			&\leq m \varepsilon_{ec},
		\end{aligned}
	\end{equation}
	where the first inequality follows the union bound and $\vec{e}_i\in\mathcal{T}^{\varepsilon}$ is the error string in the $i$th session. 
\end{proof}

Here, we note that Eve does not know the hash function $f$ before she determines the error patterns $\vec{e}_i$. That is, her choices of $\vec{e}_i$ should be independent of Alice and Bob's choice of $f\in \mathcal{F}$. Otherwise, the failure probability bound might not hold. Interestingly, this is not the case if Alice and Bob can verify the leftover errors in corrected strings, say, by exchanging an authentication tag \cite{Fung2010Finite}. Then, the failure probability is determined by the error verification process, but not the property of the hashing family. In this case, the hash function can be fixed in the beginning and known to Eve.

\section{Quantum error correction with hashing}\label{app:QECwithhash}
As mentioned in the main text, we can use Eq.~\eqref{eq:rhoABf} to further derive Eqs.~\eqref{eq:aimofQECbit} and \eqref{eq:aimofQECphase}, which are related to bit and phase error, respectively, and these two equations can be achieved by quantum error correction. We divide the procedure of quantum error correction into two steps, bit error correction to guarantee Eq.~\eqref{eq:aimofQECbit} and phase-error correction to guarantee Eq.~\eqref{eq:aimofQECphase}.

Quantum error correction can be seen as an extension of classical error correction and accomplished by using universal hashing \cite{brassard1994advances,lutkenhaus1999estimates,ma2011universally}. In the classical case, Alice and Bob each possess a bit string. The differences between these two strings are called ``errors". The main job of classical bilateral error correction is to figure out the error locations. Alice hashes her string and sends the parity information to Bob. With the same hash function, Bob hashes his string and compares it to Alice's. After enough times of this procedure, Bob can figure out the error locations and flip his corresponding bits to correct the errors. In the end, Bob's bit string is reconciled with Alice's. The number of parity-check bits is given by $n h(e)$ in the Shannon limit, where $n$ is the bit string length and $e$ is the error rate. For the case of a finite data size, there is the possibility that error correction fails. Details of error correction along with analysis of finite size effect and failure probability are presented in Appendix \ref{app:EC}.

In the following discussions, we consider linear hash functions, which can be represented by hashing matrices, for simplicity. One of the most widely used linear hash families is the family of Toeplitz matrices. The elements in a Toeplitz matrix $M$ satisfies $\forall i-j=i'-j', M_{ij}=M_{i'j'}$ and there are $m+n-1$ free bits in a Toeplitz matrix of size $m\times n$. We give an example of the Toeplitz hashing matrix as follows,
\begin{equation}\label{Eq:Hashing}
	M = \left(
	\begin{array}{cccccc}
		1 & 0 & 1 & 1 & 0 & \cdots   \\
		0 & 1 & 0 & 1 & 1 & \cdots   \\
		0 & 0 & 1 & 0 & 1 & \cdots   \\
		\vdots &  &  & \ddots &  & \cdots \\
	\end{array}
	\right)_{nh(e)\times n}.
\end{equation}
As shown before, if Alice and Bob randomly choose a Toeplitz matrix of size $nh(e)\times n$ for error correction, the efficiency converges to the Shannon limit very fast when $n$ is large.

Now, we can apply classical error correction to the quantum case. Let us start with quantum bit error correction. The parity hashing on the raw key bit strings can be implemented by a series of CNOT gates between the data qubits and ancillary EPR pairs, as shown in Figure \ref{fig:BitEC}, which is a concrete example of the quantum bit error correction part in Figure \ref{fig:QECprocedure}(a). Alice and Bob can get the parity information by measuring the ancillary qubits. The measurement result of one ancillary qubit will reflect 1 bit of parity information of the data qubits. The measurement results on the ancillary pair of Alice and Bob will be different if there is an odd number of errors in these control qubits, and the results will be the same if there is no error or an even number of errors. Alice sends the parity information to Bob, who then figures out the error syndromes and corrects the errors. The property of the universal hash family guarantees that Bob can correct all the errors with a small failure probability. A similar approach can be employed for quantum phase-error correction, with additional Hadamard gates before hash operations and measurements.

\begin{figure}[hbt!]
	\centering \includegraphics[width=15cm]{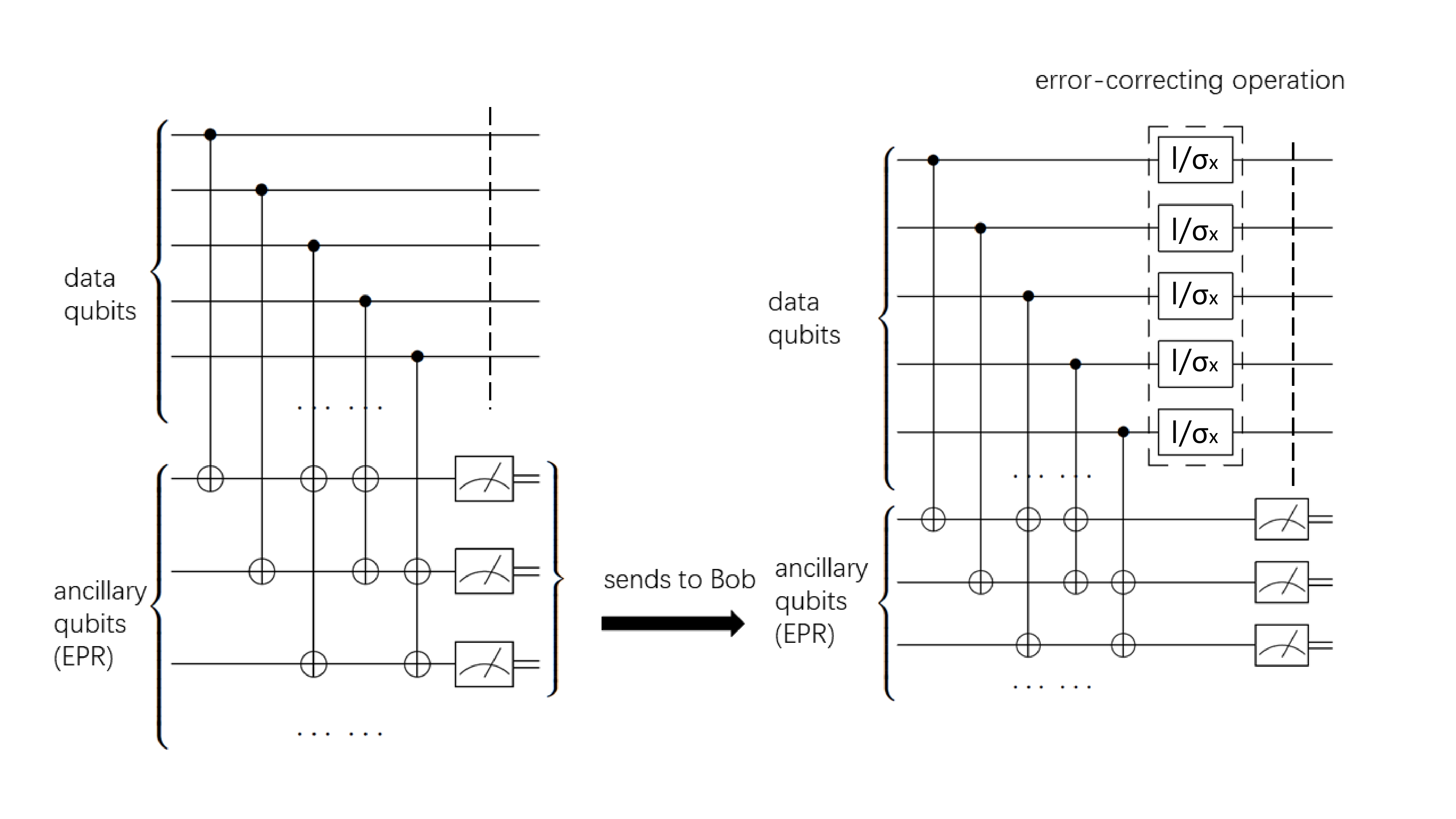}
	\caption{Illustration of quantum bit error correction. As an example, the circuit corresponds to the hashing matrix $M$ in Eq.~\eqref{Eq:Hashing}. The $i$th row of $M$ corresponds to the CNOT control data qubits targeting on the $i$th ancillary qubit, and the $j$th column of $M$ corresponds to the CNOT target ancillary qubits controlled by the $j$th data qubit. The measurement result of one ancillary qubit equals to the XOR sum of the $Z$-basis measurement results of the ancillary qubit and its controlling data qubits. By comparing the measurement results, Bob can learn 1-bit parity information of the data qubits. After knowing enough parity information on the data qubits, Bob can locate the bit errors and correct them with the quantum gate $\sigma_x$. A similar circuit can be applied for quantum phase-error correction with additional Hadamard gates according to Figure \ref{fig:QECprocedure}(a).}
	\label{fig:BitEC}
\end{figure}

In entanglement distillation, both bit and phase-error correction should be successfully implemented. The two quantum error correction procedures are carried out sequentially. Hence, we need to make sure that these two error correction procedures do not interfere with each other.
Fortunately, by using ancillary EPR pairs, we can decouple these two steps with the following lemma.

\begin{lemma}[Bit and phase-error correction decoupling \cite{lo2003method}]
	By using EPR pairs as ancillary qubits, bit error correction has no effect on phase errors and vice versa.
\end{lemma}

\begin{proof}
Let us first show that the phase-error measurement results are the same with or without the bit error correcting operations. The phase error is evaluated when both Alice and Bob perform the $X$-basis measurement on the data qubits, denoted by the measurement of the joint observable $\sigma_x\otimes\sigma_x$. From Figure \ref{fig:QECprocedure}(a), we can see that the bit error correcting operations on data qubits are essentially $I$, $\sigma_x$, and $\sigma_z$. The operations $I$ and $\sigma_x$ would not change the $X$-basis measurement outcomes. The operation $\sigma_z$ comes from the CNOT gate between the data and ancillary qubits. Since Alice's and Bob's hash functions are the same, the CNOT gates always appear in pairs. That is, if there is a CNOT between Alice's share of a data qubit pair and an ancillary EPR pair, there will be a CNOT between Bob's share of the data qubit pair and the ancillary EPR pair. From the circuit equivalency shown in Figure \ref{fig:PhaseMeasure}, we can see that the $\sigma_z$ operation always appear in pairs on Alice's and Bob's data qubits. That is, the $\sigma_z$ operation will simultaneously flip Alice's and Bob's $X$-basis measurement results, which leaves the measurement results of $\sigma_x\otimes\sigma_x$ on data qubit pairs unchanged. Therefore, quantum bit error correction does not affect phase errors.
	
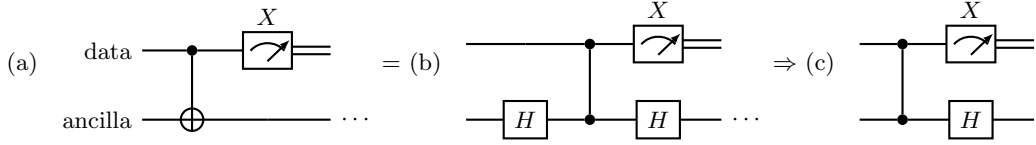
\begin{figure}[hbtp!]
		(a)
		\begin{quantikz}
			\lstick{data} \qw  & \ctrl{1}  & \meter{$X$} & \cw \\
			\lstick{ancilla} \qw  & \targ{} & \qw & \qw \rstick{$\cdots$}
		\end{quantikz}
		=
		(b)
		\begin{quantikz}
			& \qw & \ctrl{1} & \meter{$X$} & \cw \\
			& \gate{H} & \control{} & \gate{H} & \qw \rstick{$\cdots$}
		\end{quantikz}
		$\Rightarrow$
		(c)
		\begin{quantikz}
			& \ctrl{1} & \meter{$X$} & \cw \\
			& \control{} & \gate{H} & \qw
		\end{quantikz}
		\caption{A diagram showing the reason why bit error correction does not change phase errors. In the showcase, we take a pair of CNOT gates as an example. In the figures we only depict the circuit on Alice's side and the circuit is the same on Bob's side. Here, the control qubit is one of the data qubits and the target qubit is one of the ancillary qubits, as shown in Figure \ref{fig:BitEC}. The equivalence of (a), (b) and (c) comes from the facts that $\sigma_x=H\sigma_zH$ and $H^{\otimes 2}\ket{\Phi^+}=\ket{\Phi^+}$.} \label{fig:PhaseMeasure}
\end{figure}
	
With the duality between the $X$ and $Z$ bases, we can also prove that the phase-error correction would not affect bit errors with the same arguments.

\end{proof}

\section{Reduction from quantum operations to classical ones}\label{app:reduction}
Now, we want to reduce the quantum operations to classical ones, mainly following the Shor-Preskill security proof \cite{shor2000simple}. By classical, we mean that Alice and Bob can equivalently perform the operations after key measurement, which become classical data processing on key bits. The key idea is to move the final measurement before quantum error correction. Then, bit error correction becomes classical bilateral error correction and phase-error correction becomes privacy amplification. To do so, we need to make sure that the dephasing operation $\Delta_{Z^{\otimes n}}$ in Eq.~\eqref{eq:depahsingmeasure} caused by the final $Z$-basis measurement on each side commutes with all the quantum error correction operations.

Looking back at Figure \ref{fig:QECprocedure}(a), it is not difficult to see that bit error correction only consists of $I$, $\sigma_x$, and CNOT gates. These quantum operations commute with the dephasing operation $\Delta_{Z^{\otimes n}}$. Therefore, Alice and Bob can change the order of the $Z$-basis measurement and bit error correction. By doing so, quantum operations are replaced with corresponding classical operations on the measurement results, as shown in Figure \ref{fig:datareconcillation}. For example, the $\sigma_x$ gate becomes flip on classical bits and the CNOT gate becomes classical control-flip operation. Finally, quantum bit error correction becomes classical bilateral error correction, a typical information reconciliation method.

\begin{figure}[hbt!]
	\centering \includegraphics[width=12cm]{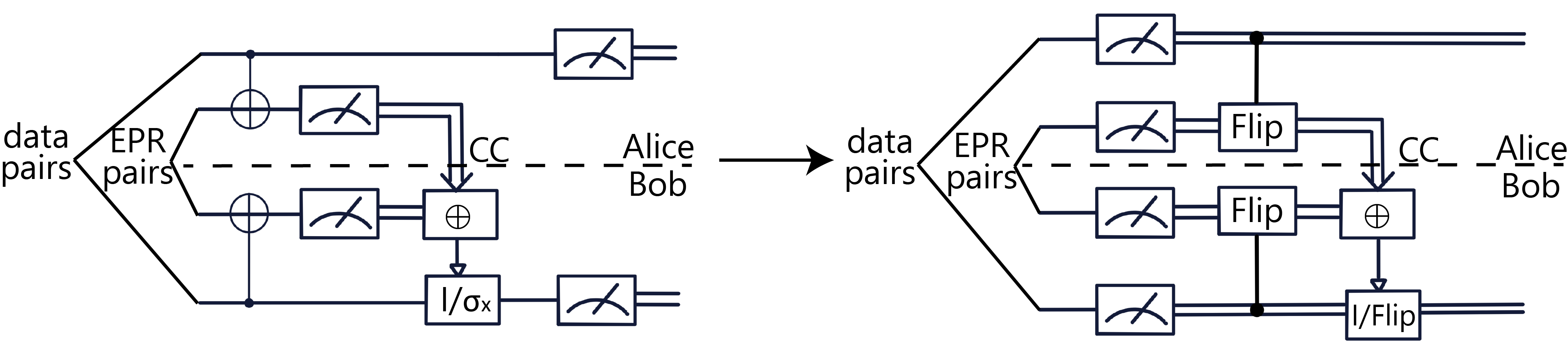}
	\caption{Transformation from quantum bit error correction to information reconciliation. The left circuit is the quantum bit error correction part in Figure \ref{fig:QECprocedure}. The ``Flip'' gate means Alice and Bob flip the corresponding classical bits. Bob calculates the error syndrome using his parity information and the one sent by Alice, and then flips his erroneous bits to reconcile the key.} \label{fig:datareconcillation}
\end{figure}

Quantum phase-error correction is more complicated because it contains one more operation, Hadamard gate, which does not commute with the dephasing operation, $\Delta_{Z^{\otimes n}}$. Here, we follow the idea of the Shor-Preskill security proof to reduce phase-error correction into classical privacy amplification \cite{shor2000simple}.

First, in the quantum circuit shown in Figure \ref{fig:QECprocedure}(a), we combine the Hadamard gate with the adjacent $Z$-basis key generation measurement together, which becomes the $X$-basis measurement. The phase error correcting operations on Bob's side is essentially $\sigma_x$, which does not affect the result of the $X$-basis measurement. Therefore, neither the correction operation nor the phase-error-syndrome measurement is necessary and the circuit of quantum phase-error correction becomes the one shown in Figure \ref{fig:QPECtoPA}(a).

\begin{figure}[hbt!]
\centering \includegraphics[width=18cm]{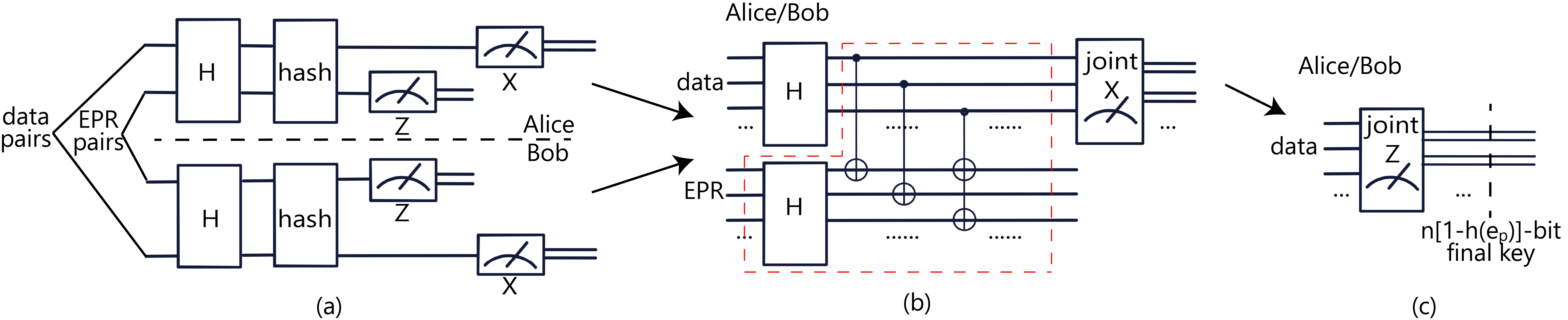}
\caption{(a) Reduced circuit from the phase-error correction part in Figure \ref{fig:QECprocedure} by considering the following two facts: Hadamard+$Z$-basis measurement = $X$-basis measurement; neither the identity nor the $\sigma_x$ gate affects the $X$-basis measurement. Now, Alice and Bob's circuits become the same. (b) Reduction of Alice's circuit: remove the measurement on ancillary qubits since the results do not affect the measurement on data qubits; replace individual $X$-basis measurements with joint $X$-basis measurements on data qubits; and explicitly express the hash operation with CNOT gates shown in Figure \ref{fig:QECprocedure}(b) as an example. If the joint $X$-basis measurements commute with the hash operation, the circuit in the red dashed box will not affect the measurements and hence can be removed. (c) Further reduction of Alice's circuit: remove the redundant circuit in the red dashed box and combine joint $X$-basis measurements with the Hadamard gates which become joint $Z$-basis measurements. Finally, Alice can get an $n[1-h(e_p)]$-bit secure key from the joint $Z$-basis measurements in a QKD session.} \label{fig:QPECtoPA}
\end{figure}

Second, since Alice and Bob's positions are symmetric now, we focus on Alice's side. We shall show that she can replace the individual key measurement with a series of joint $X$-basis measurements and still get an $[n-nh(e_p)]$-bit secure final key, and as shown in Figure \ref{fig:QPECtoPA}(b). The joint measurement can be represented by an $n$-bit vector $\vec{v}$, which measures the observable
\begin{equation}\label{eq:jointX}
O_{\vec{v},x} = \bigotimes_{i=1}^n \sigma_x^{v_{i}},
\end{equation}
where $v_i$ is the element values of $\vec{v}$. 
Here, we require $\vec{v}$s to be linearly independent so that the final key is still secure after phase-error correction. Note that since the key measurements differ from the ones in Figure~\ref{fig:QPECtoPA}(a), the obtained key bits may be different.

The effect of hash operations on the $n$ data qubits is either identity when the measurement outcome of the ancillary qubit is 0 or a serious of $n$-qubit operations consisting of $\sigma_z$ and $I$ when the outcome is 1. For example, the hash operation corresponding to the first row vector $\vec{m}_1$ of the matrix in Eq.~\eqref{Eq:Hashing} is an identity or
\begin{equation}
O_{\vec{m}_1,z}=\sigma_z \otimes I \otimes \sigma_z \otimes \sigma_z \otimes I \cdots.
\end{equation}
In general, we represent the operation corresponding to the $i$th row vector $\vec{m}_i$ of the matrix $M$ as follows,
\begin{equation} \label{eq:hashZ}
O_{\vec{m}_i,z}=\bigotimes_{j=1}^n \sigma_z^{M_{ij}},
\end{equation}
where $M_{ij}$ is the value of the element in the $i$th row and the $j$th column of the matrix. 

Normally, the operators in Eqs.~\eqref{eq:jointX} and \eqref{eq:hashZ} do not commute because $[\sigma_x,\sigma_z]\neq0$. Fortunately, we have $[\sigma_x\otimes\sigma_x,\sigma_z\otimes\sigma_z]=0$. To make sure that a joint $X$-basis measurement commutes with a serious of $Z$ operations, we only need to design the joint $X$-basis measurement such that the number of qubits having been applied with the $\sigma_z$ operation and measured in the $X$ basis is even. That is, $\vec{m}_i\cdot \vec{v}=0$ holds, $\forall i\in\{1,2,3,...,nh(e_p)\}$. This is equivalent to finding the kernel of $M$. The kernel of the matrix $M$ is defined as the set of vectors such that
\begin{equation}
\ker{M}=\{\vec{v}:M\vec{v}=\vec{0}\}.
\end{equation}
In phase-error correction, the hashing matrix we use has a full rank. Therefore, the rank of $\ker{M}$ is $n-nh(e_p)$, or, $\ker{M}$ can be constructed from $[n-nh(e_p)]$ linearly independent vectors. We arrange these vectors in columns and form a matrix, $V$, which is of the size $n\times[n-nh(e_p)]$. We call the matrix $V$ the dual matrix of $M$. Then, we can design joint $X$-basis measurements according to this dual matrix, where each joint measurement corresponds to a column vector of $V$ according to the correspondence in Eq.~\eqref{eq:jointX}. By construction, these joint measurements all commute with the hash operation.

Finally, the commuting property shows that the hash operation would not affect the results of the joint $X$-basis measurements. Then, Alice and Bob can remove the hash operation along with the ancillary qubits. The joint $X$-basis measurements can be further combined with the Hadamard gate and become the joint $Z$-basis measurements, as shown in Figure \ref{fig:QPECtoPA}(c). The measurement results of the joint $Z$-basis measurements give an $[n-nh(e_p)]$-bit secret key.

With the reduction, one can see that the ancillary EPR pairs are unnecessary for phase-error correction, as shown in Figure \ref{fig:QPECtoPA}(c). Does that mean, there is no cost of shared private randomness in phase-error correction? Unfortunately, the answer is no. The cost is reflected in the joint $Z$-basis measurement. The number of final key bits is determined by the number of joint $Z$-basis measurements, which is limited by the kernel of the hashing matrix for phase-error correction.

Meanwhile, the joint $Z$-basis measurement is compatible with bit error correction reduction. After considering the cost of shared private randomness in bit error correction, the net gain of secret key is $n[1-h(e_b)-h(e_p)]$ bits, which matches the key rate in Eq.~\eqref{eq:ShorPreskillR}.

There are a few notes on information reconciliation and privacy amplification.
\begin{enumerate}
\item
We take the family of Toeplitz matrices as an example here. The decoding for such random-Toeplitz-matrix hashing is computationally hard in practice. To bypass the hardness in information reconciliation, people use more practical error correcting codes, such as the low-density parity-check code \cite{gallager1962low,mackay1999good}. For privacy amplification, Alice and Bob can still adopt random-Toeplitz-matrix hashing, as the decoding is unnecessary for phase-error correction.
	
\item
In general, the failure probability of error correction $\varepsilon$ is a property of the hash family. But if Alice and Bob can perform error verification, the failure probability of error correction is determined by the verification process instead \cite{Fung2010Finite}. In this case, the hash function can be pre-fixed and known to Eve in the beginning. This is the case for information reconciliation. Unfortunately, there is no way to verify the result of privacy amplification because the virtual quantum phase-error correction is not implemented in real life.
	
\item
The reduction from quantum phase-error correction to privacy amplification is different from Shor-Preskill's original argument, where the Calderbank-Shor-Steane quantum error-correcting code is adopted.
\end{enumerate}

\section{GLLP framework}\label{app:examplewithGLLP}
In practical implementation, the devices may deviate from the ideal assumptions in the security analysis. The actual error rate may also vary from the prior estimation from random sampling due to statistical fluctuations. We need to take these issues into account in the security analysis, where they affect the secure key rate given in Eq.~\eqref{eq:ShorPreskillR}. Thus, we also need to modify the parameter settings of the stream privacy amplification scheme accordingly in practice. Here, we take the GLLP framework as an example to show how the stream scheme can be combined with existing approaches to handle the practical issues.

In Ref.~\cite{gottesman2004security}, Gottesman et al.~established a general framework for security analysis with realistic devices. In this framework, Alice and Bob characterize their devices to quantify the deviation from the ideal case. For this purpose, Alice and Bob can perform a virtual measurement on the devices, and then for each round, they tag the sifted key bit as ``good'' if the virtual measurement projects to the ideal case and ``bad'' if it is the orthogonal case. More generally, Alice and Bob can label sifted key bits with an arbitrary tag $g$ and derive the corresponding phase error rate $e^g_p$. Then we can get the extended GLLP key-rate formula~\cite{Ma2008PhD},
\begin{equation}\label{Eq:GLLP}
	r\geq -h(e_b)+\sum_gq_g\left[1-h(e^g_p)\right],
\end{equation}
where $e_b$ is the bit error rate and $\{q_g\}$ gives the ratio of the sifted key bits with tag $g$, satisfying $\sum_gq_g=1$ and $g_q\geq0,\forall g$. The first term on the right hand side of Eq.~\eqref{Eq:GLLP} represents the cost of information reconciliation and the second term corresponds to the ratio of privacy amplification.

With this formula, we can directly employ stream privacy amplification to the GLLP framework. Here, we still consider the case of performing information reconciliation first and then privacy amplification in the postprocessing. We denote the length of the sifted key string as $n$. If Alice and Bob perform encryption in information reconciliation, then in stream privacy amplification, the seed length is $n\sum_gq_g h(e^g_p)$ and the size of the hashing matrix is $\left[n\sum_gq_gh(e^g_p)\right]\times n$. As for the case where Alice and Bob do not perform encryption in information reconciliation, the seed length will be $n\left[h(e_b)+\sum_gq_g h(e^g_p)\right]$ and the size of the hashing matrix will be $\left[nh(e_b)+n\sum_gq_gh(e^g_p)\right]\times n$ in stream privacy amplification. The rest steps are the same as those in Box~\ref{box:streamPA}. 

As we can see from this example, the new scheme is highly portable. Given a QKD protocol in the presence of practical issues, once the users can analyze the amount of information leakage and obtain a good estimation, they can carry out privacy amplification in a stream way. Thus, the stream scheme is adequate to deal with practical issues.



\bibliographystyle{apsrev}

\bibliography{bibIntroQC}

\begin{thebibliography}{52}
\expandafter\ifx\csname natexlab\endcsname\relax\def\natexlab#1{#1}\fi
\expandafter\ifx\csname bibnamefont\endcsname\relax
  \def\bibnamefont#1{#1}\fi
\expandafter\ifx\csname bibfnamefont\endcsname\relax
  \def\bibfnamefont#1{#1}\fi
\expandafter\ifx\csname citenamefont\endcsname\relax
  \def\citenamefont#1{#1}\fi
\expandafter\ifx\csname url\endcsname\relax
  \def\url#1{\texttt{#1}}\fi
\expandafter\ifx\csname urlprefix\endcsname\relax\def\urlprefix{URL }\fi
\providecommand{\bibinfo}[2]{#2}
\providecommand{\eprint}[2][]{\url{#2}}

\bibitem[{\citenamefont{Bennett and Brassard}(1984)}]{bennett1984proceedings}
\bibinfo{author}{\bibfnamefont{C.~H.} \bibnamefont{Bennett}} \bibnamefont{and}
  \bibinfo{author}{\bibfnamefont{G.}~\bibnamefont{Brassard}}, in
  \emph{\bibinfo{booktitle}{Proceedings of the IEEE International Conference on
  Computers, Systems and Signal Processing}} (\bibinfo{publisher}{IEEE Press},
  \bibinfo{address}{New York}, \bibinfo{year}{1984}), pp.
  \bibinfo{pages}{175--179},
  \urlprefix\url{https://doi.org/10.1016/j.tcs.2014.05.025}.

\bibitem[{\citenamefont{Ekert}(1991)}]{ekert1991quantum}
\bibinfo{author}{\bibfnamefont{A.~K.} \bibnamefont{Ekert}},
  \bibinfo{journal}{Phys. Rev. Lett.} \textbf{\bibinfo{volume}{67}},
  \bibinfo{pages}{661} (\bibinfo{year}{1991}),
  \urlprefix\url{https://link.aps.org/doi/10.1103/PhysRevLett.67.661}.

\bibitem[{\citenamefont{Bennett et~al.}(1988)\citenamefont{Bennett, Brassard,
  and Robert}}]{bennett1988privacy}
\bibinfo{author}{\bibfnamefont{C.~H.} \bibnamefont{Bennett}},
  \bibinfo{author}{\bibfnamefont{G.}~\bibnamefont{Brassard}}, \bibnamefont{and}
  \bibinfo{author}{\bibfnamefont{J.-M.} \bibnamefont{Robert}},
  \bibinfo{journal}{SIAM J. Comput.} \textbf{\bibinfo{volume}{17}},
  \bibinfo{pages}{210} (\bibinfo{year}{1988}),
  \urlprefix\url{https://doi.org/10.1137/0217014}.

\bibitem[{\citenamefont{Maurer}(1993)}]{maurer1993secret}
\bibinfo{author}{\bibfnamefont{U.}~\bibnamefont{Maurer}},
  \bibinfo{journal}{IEEE Trans. Inf. Theory} \textbf{\bibinfo{volume}{39}},
  \bibinfo{pages}{733} (\bibinfo{year}{1993}).

\bibitem[{\citenamefont{Xu et~al.}(2020)\citenamefont{Xu, Ma, Zhang, Lo, and
  Pan}}]{xu2020secure}
\bibinfo{author}{\bibfnamefont{F.}~\bibnamefont{Xu}},
  \bibinfo{author}{\bibfnamefont{X.}~\bibnamefont{Ma}},
  \bibinfo{author}{\bibfnamefont{Q.}~\bibnamefont{Zhang}},
  \bibinfo{author}{\bibfnamefont{H.-K.} \bibnamefont{Lo}}, \bibnamefont{and}
  \bibinfo{author}{\bibfnamefont{J.-W.} \bibnamefont{Pan}},
  \bibinfo{journal}{Rev. Mod. Phys.} \textbf{\bibinfo{volume}{92}},
  \bibinfo{pages}{025002} (\bibinfo{year}{2020}),
  \urlprefix\url{https://link.aps.org/doi/10.1103/RevModPhys.92.025002}.

\bibitem[{\citenamefont{Chen et~al.}(2020)\citenamefont{Chen, Zhang, Liu,
  Jiang, Zhang, Hu, Guan, Yu, Xu, Lin et~al.}}]{chen2020sending}
\bibinfo{author}{\bibfnamefont{J.-P.} \bibnamefont{Chen}},
  \bibinfo{author}{\bibfnamefont{C.}~\bibnamefont{Zhang}},
  \bibinfo{author}{\bibfnamefont{Y.}~\bibnamefont{Liu}},
  \bibinfo{author}{\bibfnamefont{C.}~\bibnamefont{Jiang}},
  \bibinfo{author}{\bibfnamefont{W.}~\bibnamefont{Zhang}},
  \bibinfo{author}{\bibfnamefont{X.-L.} \bibnamefont{Hu}},
  \bibinfo{author}{\bibfnamefont{J.-Y.} \bibnamefont{Guan}},
  \bibinfo{author}{\bibfnamefont{Z.-W.} \bibnamefont{Yu}},
  \bibinfo{author}{\bibfnamefont{H.}~\bibnamefont{Xu}},
  \bibinfo{author}{\bibfnamefont{J.}~\bibnamefont{Lin}}, \bibnamefont{et~al.},
  \bibinfo{journal}{Phys. Rev. Lett.} \textbf{\bibinfo{volume}{124}},
  \bibinfo{pages}{070501} (\bibinfo{year}{2020}),
  \urlprefix\url{https://link.aps.org/doi/10.1103/PhysRevLett.124.070501}.

\bibitem[{\citenamefont{Fang et~al.}(2020)\citenamefont{Fang, Zeng, Liu, Zou,
  Wu, Tang, Sheng, Xiang, Zhang, Li et~al.}}]{fang2020implementation}
\bibinfo{author}{\bibfnamefont{X.-T.} \bibnamefont{Fang}},
  \bibinfo{author}{\bibfnamefont{P.}~\bibnamefont{Zeng}},
  \bibinfo{author}{\bibfnamefont{H.}~\bibnamefont{Liu}},
  \bibinfo{author}{\bibfnamefont{M.}~\bibnamefont{Zou}},
  \bibinfo{author}{\bibfnamefont{W.}~\bibnamefont{Wu}},
  \bibinfo{author}{\bibfnamefont{Y.-L.} \bibnamefont{Tang}},
  \bibinfo{author}{\bibfnamefont{Y.-J.} \bibnamefont{Sheng}},
  \bibinfo{author}{\bibfnamefont{Y.}~\bibnamefont{Xiang}},
  \bibinfo{author}{\bibfnamefont{W.}~\bibnamefont{Zhang}},
  \bibinfo{author}{\bibfnamefont{H.}~\bibnamefont{Li}}, \bibnamefont{et~al.},
  \bibinfo{journal}{Nature Photonics} \textbf{\bibinfo{volume}{14}},
  \bibinfo{pages}{422} (\bibinfo{year}{2020}).

\bibitem[{\citenamefont{Liao et~al.}(2018)\citenamefont{Liao, Cai, Handsteiner,
  Liu, Yin, Zhang, Rauch, Fink, Ren, Liu et~al.}}]{liao2018satellite}
\bibinfo{author}{\bibfnamefont{S.-K.} \bibnamefont{Liao}},
  \bibinfo{author}{\bibfnamefont{W.-Q.} \bibnamefont{Cai}},
  \bibinfo{author}{\bibfnamefont{J.}~\bibnamefont{Handsteiner}},
  \bibinfo{author}{\bibfnamefont{B.}~\bibnamefont{Liu}},
  \bibinfo{author}{\bibfnamefont{J.}~\bibnamefont{Yin}},
  \bibinfo{author}{\bibfnamefont{L.}~\bibnamefont{Zhang}},
  \bibinfo{author}{\bibfnamefont{D.}~\bibnamefont{Rauch}},
  \bibinfo{author}{\bibfnamefont{M.}~\bibnamefont{Fink}},
  \bibinfo{author}{\bibfnamefont{J.-G.} \bibnamefont{Ren}},
  \bibinfo{author}{\bibfnamefont{W.-Y.} \bibnamefont{Liu}},
  \bibnamefont{et~al.}, \bibinfo{journal}{Phys. Rev. Lett.}
  \textbf{\bibinfo{volume}{120}}, \bibinfo{pages}{030501}
  (\bibinfo{year}{2018}),
  \urlprefix\url{https://link.aps.org/doi/10.1103/PhysRevLett.120.030501}.

\bibitem[{\citenamefont{Islam et~al.}(2017)\citenamefont{Islam, Lim, Cahall,
  Kim, and Gauthier}}]{islam2017provably}
\bibinfo{author}{\bibfnamefont{N.~T.} \bibnamefont{Islam}},
  \bibinfo{author}{\bibfnamefont{C.~C.~W.} \bibnamefont{Lim}},
  \bibinfo{author}{\bibfnamefont{C.}~\bibnamefont{Cahall}},
  \bibinfo{author}{\bibfnamefont{J.}~\bibnamefont{Kim}}, \bibnamefont{and}
  \bibinfo{author}{\bibfnamefont{D.~J.} \bibnamefont{Gauthier}},
  \bibinfo{journal}{Sci. Adv.} \textbf{\bibinfo{volume}{3}}
  (\bibinfo{year}{2017}),
  \urlprefix\url{https://advances.sciencemag.org/content/3/11/e1701491}.

\bibitem[{\citenamefont{Elliott et~al.}(2005)\citenamefont{Elliott, Colvin,
  Pearson, Pikalo, Schlafer, and Yeh}}]{elliott2005current}
\bibinfo{author}{\bibfnamefont{C.}~\bibnamefont{Elliott}},
  \bibinfo{author}{\bibfnamefont{A.}~\bibnamefont{Colvin}},
  \bibinfo{author}{\bibfnamefont{D.}~\bibnamefont{Pearson}},
  \bibinfo{author}{\bibfnamefont{O.}~\bibnamefont{Pikalo}},
  \bibinfo{author}{\bibfnamefont{J.}~\bibnamefont{Schlafer}}, \bibnamefont{and}
  \bibinfo{author}{\bibfnamefont{H.}~\bibnamefont{Yeh}}, in
  \emph{\bibinfo{booktitle}{Quantum Information and computation III}}
  (\bibinfo{organization}{International Society for Optics and Photonics},
  \bibinfo{year}{2005}), vol. \bibinfo{volume}{5815}, pp.
  \bibinfo{pages}{138--149}, \urlprefix\url{https://doi.org/10.1117/12.606489}.

\bibitem[{\citenamefont{Peev et~al.}(2009)\citenamefont{Peev, Pacher,
  All{\'{e}}aume, Barreiro, Bouda, Boxleitner, Debuisschert, Diamanti, Dianati,
  Dynes et~al.}}]{peev2009secoqc}
\bibinfo{author}{\bibfnamefont{M.}~\bibnamefont{Peev}},
  \bibinfo{author}{\bibfnamefont{C.}~\bibnamefont{Pacher}},
  \bibinfo{author}{\bibfnamefont{R.}~\bibnamefont{All{\'{e}}aume}},
  \bibinfo{author}{\bibfnamefont{C.}~\bibnamefont{Barreiro}},
  \bibinfo{author}{\bibfnamefont{J.}~\bibnamefont{Bouda}},
  \bibinfo{author}{\bibfnamefont{W.}~\bibnamefont{Boxleitner}},
  \bibinfo{author}{\bibfnamefont{T.}~\bibnamefont{Debuisschert}},
  \bibinfo{author}{\bibfnamefont{E.}~\bibnamefont{Diamanti}},
  \bibinfo{author}{\bibfnamefont{M.}~\bibnamefont{Dianati}},
  \bibinfo{author}{\bibfnamefont{J.~F.} \bibnamefont{Dynes}},
  \bibnamefont{et~al.}, \bibinfo{journal}{New J. Phys.}
  \textbf{\bibinfo{volume}{11}}, \bibinfo{pages}{075001}
  (\bibinfo{year}{2009}),
  \urlprefix\url{https://doi.org/10.1088/1367-2630/11/7/075001}.

\bibitem[{\citenamefont{Stucki et~al.}(2011)\citenamefont{Stucki, Legre,
  Buntschu, Clausen, Felber, Gisin, Henzen, Junod, Litzistorf, Monbaron
  et~al.}}]{stucki2011long}
\bibinfo{author}{\bibfnamefont{D.}~\bibnamefont{Stucki}},
  \bibinfo{author}{\bibfnamefont{M.}~\bibnamefont{Legre}},
  \bibinfo{author}{\bibfnamefont{F.}~\bibnamefont{Buntschu}},
  \bibinfo{author}{\bibfnamefont{B.}~\bibnamefont{Clausen}},
  \bibinfo{author}{\bibfnamefont{N.}~\bibnamefont{Felber}},
  \bibinfo{author}{\bibfnamefont{N.}~\bibnamefont{Gisin}},
  \bibinfo{author}{\bibfnamefont{L.}~\bibnamefont{Henzen}},
  \bibinfo{author}{\bibfnamefont{P.}~\bibnamefont{Junod}},
  \bibinfo{author}{\bibfnamefont{G.}~\bibnamefont{Litzistorf}},
  \bibinfo{author}{\bibfnamefont{P.}~\bibnamefont{Monbaron}},
  \bibnamefont{et~al.}, \bibinfo{journal}{New J. Phys.}
  \textbf{\bibinfo{volume}{13}}, \bibinfo{pages}{123001}
  (\bibinfo{year}{2011}),
  \urlprefix\url{https://doi.org/10.1088/1367-2630/13/12/123001}.

\bibitem[{\citenamefont{Sasaki et~al.}(2011)\citenamefont{Sasaki, Fujiwara,
  Ishizuka, Klaus, Wakui, Takeoka, Miki, Yamashita, Wang, Tanaka
  et~al.}}]{sasaki2011field}
\bibinfo{author}{\bibfnamefont{M.}~\bibnamefont{Sasaki}},
  \bibinfo{author}{\bibfnamefont{M.}~\bibnamefont{Fujiwara}},
  \bibinfo{author}{\bibfnamefont{H.}~\bibnamefont{Ishizuka}},
  \bibinfo{author}{\bibfnamefont{W.}~\bibnamefont{Klaus}},
  \bibinfo{author}{\bibfnamefont{K.}~\bibnamefont{Wakui}},
  \bibinfo{author}{\bibfnamefont{M.}~\bibnamefont{Takeoka}},
  \bibinfo{author}{\bibfnamefont{S.}~\bibnamefont{Miki}},
  \bibinfo{author}{\bibfnamefont{T.}~\bibnamefont{Yamashita}},
  \bibinfo{author}{\bibfnamefont{Z.}~\bibnamefont{Wang}},
  \bibinfo{author}{\bibfnamefont{A.}~\bibnamefont{Tanaka}},
  \bibnamefont{et~al.}, \bibinfo{journal}{Opt. Express}
  \textbf{\bibinfo{volume}{19}}, \bibinfo{pages}{10387} (\bibinfo{year}{2011}),
  \urlprefix\url{http://www.opticsexpress.org/abstract.cfm?URI=oe-19-11-10387}.

\bibitem[{\citenamefont{Chen et~al.}(2009)\citenamefont{Chen, Liang, Liu, Cai,
  Ju, Liu, Wang, Yin, Chen, Chen et~al.}}]{chen2009field}
\bibinfo{author}{\bibfnamefont{T.-Y.} \bibnamefont{Chen}},
  \bibinfo{author}{\bibfnamefont{H.}~\bibnamefont{Liang}},
  \bibinfo{author}{\bibfnamefont{Y.}~\bibnamefont{Liu}},
  \bibinfo{author}{\bibfnamefont{W.-Q.} \bibnamefont{Cai}},
  \bibinfo{author}{\bibfnamefont{L.}~\bibnamefont{Ju}},
  \bibinfo{author}{\bibfnamefont{W.-Y.} \bibnamefont{Liu}},
  \bibinfo{author}{\bibfnamefont{J.}~\bibnamefont{Wang}},
  \bibinfo{author}{\bibfnamefont{H.}~\bibnamefont{Yin}},
  \bibinfo{author}{\bibfnamefont{K.}~\bibnamefont{Chen}},
  \bibinfo{author}{\bibfnamefont{Z.-B.} \bibnamefont{Chen}},
  \bibnamefont{et~al.}, \bibinfo{journal}{Opt. Express}
  \textbf{\bibinfo{volume}{17}}, \bibinfo{pages}{6540} (\bibinfo{year}{2009}),
  \urlprefix\url{http://www.opticsexpress.org/abstract.cfm?URI=oe-17-8-6540}.

\bibitem[{\citenamefont{Wang et~al.}(2010)\citenamefont{Wang, Chen, Yin, Zhang,
  Zhang, Li, Xu, Zhou, Yang, Huang et~al.}}]{wang2010field}
\bibinfo{author}{\bibfnamefont{S.}~\bibnamefont{Wang}},
  \bibinfo{author}{\bibfnamefont{W.}~\bibnamefont{Chen}},
  \bibinfo{author}{\bibfnamefont{Z.-Q.} \bibnamefont{Yin}},
  \bibinfo{author}{\bibfnamefont{Y.}~\bibnamefont{Zhang}},
  \bibinfo{author}{\bibfnamefont{T.}~\bibnamefont{Zhang}},
  \bibinfo{author}{\bibfnamefont{H.-W.} \bibnamefont{Li}},
  \bibinfo{author}{\bibfnamefont{F.-X.} \bibnamefont{Xu}},
  \bibinfo{author}{\bibfnamefont{Z.}~\bibnamefont{Zhou}},
  \bibinfo{author}{\bibfnamefont{Y.}~\bibnamefont{Yang}},
  \bibinfo{author}{\bibfnamefont{D.-J.} \bibnamefont{Huang}},
  \bibnamefont{et~al.}, \bibinfo{journal}{Opt. Lett.}
  \textbf{\bibinfo{volume}{35}}, \bibinfo{pages}{2454} (\bibinfo{year}{2010}),
  \urlprefix\url{http://ol.osa.org/abstract.cfm?URI=ol-35-14-2454}.

\bibitem[{\citenamefont{Chen et~al.}(2021{\natexlab{a}})\citenamefont{Chen,
  Zhang, Chen, Cai, Liao, Zhang, Chen, Yin, Ren, Chen
  et~al.}}]{chen2021integrated}
\bibinfo{author}{\bibfnamefont{Y.-A.} \bibnamefont{Chen}},
  \bibinfo{author}{\bibfnamefont{Q.}~\bibnamefont{Zhang}},
  \bibinfo{author}{\bibfnamefont{T.-Y.} \bibnamefont{Chen}},
  \bibinfo{author}{\bibfnamefont{W.-Q.} \bibnamefont{Cai}},
  \bibinfo{author}{\bibfnamefont{S.-K.} \bibnamefont{Liao}},
  \bibinfo{author}{\bibfnamefont{J.}~\bibnamefont{Zhang}},
  \bibinfo{author}{\bibfnamefont{K.}~\bibnamefont{Chen}},
  \bibinfo{author}{\bibfnamefont{J.}~\bibnamefont{Yin}},
  \bibinfo{author}{\bibfnamefont{J.-G.} \bibnamefont{Ren}},
  \bibinfo{author}{\bibfnamefont{Z.}~\bibnamefont{Chen}}, \bibnamefont{et~al.},
  \bibinfo{journal}{Nature} \textbf{\bibinfo{volume}{589}},
  \bibinfo{pages}{214} (\bibinfo{year}{2021}{\natexlab{a}}),
  \urlprefix\url{https://www.nature.com/articles/s41586-020-03093-8}.

\bibitem[{\citenamefont{Qiu}(2014)}]{qiu2014quantum}
\bibinfo{author}{\bibfnamefont{J.}~\bibnamefont{Qiu}}, \bibinfo{journal}{Nature
  News} \textbf{\bibinfo{volume}{508}}, \bibinfo{pages}{441}
  (\bibinfo{year}{2014}),
  \urlprefix\url{https://www.nature.com/articles/508441a}.

\bibitem[{\citenamefont{Fung et~al.}(2010)\citenamefont{Fung, Ma, and
  Chau}}]{Fung2010Finite}
\bibinfo{author}{\bibfnamefont{C.-H.~F.} \bibnamefont{Fung}},
  \bibinfo{author}{\bibfnamefont{X.}~\bibnamefont{Ma}}, \bibnamefont{and}
  \bibinfo{author}{\bibfnamefont{H.~F.} \bibnamefont{Chau}},
  \bibinfo{journal}{Phys. Rev. A} \textbf{\bibinfo{volume}{81}},
  \bibinfo{pages}{012318} (\bibinfo{year}{2010}),
  \urlprefix\url{https://link.aps.org/doi/10.1103/PhysRevA.81.012318}.

\bibitem[{\citenamefont{Liao et~al.}(2017)\citenamefont{Liao, Lin, Ren, Liu,
  Qiang, Yin, Li, Shen, Zhang, Liang et~al.}}]{Liao2017Space}
\bibinfo{author}{\bibfnamefont{S.-K.} \bibnamefont{Liao}},
  \bibinfo{author}{\bibfnamefont{J.}~\bibnamefont{Lin}},
  \bibinfo{author}{\bibfnamefont{J.-G.} \bibnamefont{Ren}},
  \bibinfo{author}{\bibfnamefont{W.-Y.} \bibnamefont{Liu}},
  \bibinfo{author}{\bibfnamefont{J.}~\bibnamefont{Qiang}},
  \bibinfo{author}{\bibfnamefont{J.}~\bibnamefont{Yin}},
  \bibinfo{author}{\bibfnamefont{Y.}~\bibnamefont{Li}},
  \bibinfo{author}{\bibfnamefont{Q.}~\bibnamefont{Shen}},
  \bibinfo{author}{\bibfnamefont{L.}~\bibnamefont{Zhang}},
  \bibinfo{author}{\bibfnamefont{X.-F.} \bibnamefont{Liang}},
  \bibnamefont{et~al.}, \bibinfo{journal}{Chinese Phys. Lett.}
  \textbf{\bibinfo{volume}{34}}, \bibinfo{eid}{090302}
  (pages~\bibinfo{numpages}{0}) (\bibinfo{year}{2017}),
  \urlprefix\url{http://cpl.iphy.ac.cn/EN/abstract/article_70915.shtml}.

\bibitem[{\citenamefont{Liu et~al.}(2018)\citenamefont{Liu, Zhao, Li, Guan,
  Zhang, Bai, Zhang, Liu, Wu, Yuan et~al.}}]{liu2018device}
\bibinfo{author}{\bibfnamefont{Y.}~\bibnamefont{Liu}},
  \bibinfo{author}{\bibfnamefont{Q.}~\bibnamefont{Zhao}},
  \bibinfo{author}{\bibfnamefont{M.-H.} \bibnamefont{Li}},
  \bibinfo{author}{\bibfnamefont{J.-Y.} \bibnamefont{Guan}},
  \bibinfo{author}{\bibfnamefont{Y.}~\bibnamefont{Zhang}},
  \bibinfo{author}{\bibfnamefont{B.}~\bibnamefont{Bai}},
  \bibinfo{author}{\bibfnamefont{W.}~\bibnamefont{Zhang}},
  \bibinfo{author}{\bibfnamefont{W.-Z.} \bibnamefont{Liu}},
  \bibinfo{author}{\bibfnamefont{C.}~\bibnamefont{Wu}},
  \bibinfo{author}{\bibfnamefont{X.}~\bibnamefont{Yuan}}, \bibnamefont{et~al.},
  \bibinfo{journal}{Nature} \textbf{\bibinfo{volume}{562}},
  \bibinfo{pages}{548} (\bibinfo{year}{2018}),
  \urlprefix\url{https://www.nature.com/articles/s41586%20018%200559%203}.

\bibitem[{\citenamefont{Bennett et~al.}(1995)\citenamefont{Bennett, Brassard,
  Cr{\'e}peau, and Maurer}}]{bennett1995generalized}
\bibinfo{author}{\bibfnamefont{C.}~\bibnamefont{Bennett}},
  \bibinfo{author}{\bibfnamefont{G.}~\bibnamefont{Brassard}},
  \bibinfo{author}{\bibfnamefont{C.}~\bibnamefont{Cr{\'e}peau}},
  \bibnamefont{and} \bibinfo{author}{\bibfnamefont{U.}~\bibnamefont{Maurer}},
  \bibinfo{journal}{IEEE Trans. Inf. Theory} \textbf{\bibinfo{volume}{41}},
  \bibinfo{pages}{1915} (\bibinfo{year}{1995}),
  \urlprefix\url{https://doi.org/10.1109/18.476316}.

\bibitem[{\citenamefont{Ma et~al.}(2013)\citenamefont{Ma, Xu, Xu, Tan, Qi, and
  Lo}}]{Ma2013Extractor}
\bibinfo{author}{\bibfnamefont{X.}~\bibnamefont{Ma}},
  \bibinfo{author}{\bibfnamefont{F.}~\bibnamefont{Xu}},
  \bibinfo{author}{\bibfnamefont{H.}~\bibnamefont{Xu}},
  \bibinfo{author}{\bibfnamefont{X.}~\bibnamefont{Tan}},
  \bibinfo{author}{\bibfnamefont{B.}~\bibnamefont{Qi}}, \bibnamefont{and}
  \bibinfo{author}{\bibfnamefont{H.-K.} \bibnamefont{Lo}},
  \bibinfo{journal}{Phys. Rev. A} \textbf{\bibinfo{volume}{87}},
  \bibinfo{pages}{062327} (\bibinfo{year}{2013}),
  \urlprefix\url{http://link.aps.org/doi/10.1103/PhysRevA.87.062327}.

\bibitem[{\citenamefont{Hayashi and Tsurumaru}(2016)}]{hayashi2016more}
\bibinfo{author}{\bibfnamefont{M.}~\bibnamefont{Hayashi}} \bibnamefont{and}
  \bibinfo{author}{\bibfnamefont{T.}~\bibnamefont{Tsurumaru}},
  \bibinfo{journal}{IEEE Trans. Inf. Theory} \textbf{\bibinfo{volume}{62}},
  \bibinfo{pages}{2213} (\bibinfo{year}{2016}),
  \urlprefix\url{https://doi.org/10.1109/TIT.2016.2526018}.

\bibitem[{\citenamefont{Chen et~al.}(2021{\natexlab{b}})\citenamefont{Chen,
  Jiang, Tang, Zhou, Yuan, Zhou, Wang, Liu, Chen, Liu
  et~al.}}]{Chen2021implementation}
\bibinfo{author}{\bibfnamefont{T.-Y.} \bibnamefont{Chen}},
  \bibinfo{author}{\bibfnamefont{X.}~\bibnamefont{Jiang}},
  \bibinfo{author}{\bibfnamefont{S.-B.} \bibnamefont{Tang}},
  \bibinfo{author}{\bibfnamefont{L.}~\bibnamefont{Zhou}},
  \bibinfo{author}{\bibfnamefont{X.}~\bibnamefont{Yuan}},
  \bibinfo{author}{\bibfnamefont{H.}~\bibnamefont{Zhou}},
  \bibinfo{author}{\bibfnamefont{J.}~\bibnamefont{Wang}},
  \bibinfo{author}{\bibfnamefont{Y.}~\bibnamefont{Liu}},
  \bibinfo{author}{\bibfnamefont{L.-K.} \bibnamefont{Chen}},
  \bibinfo{author}{\bibfnamefont{W.-Y.} \bibnamefont{Liu}},
  \bibnamefont{et~al.}, \bibinfo{journal}{NPJ Quantum Inf.}
  \textbf{\bibinfo{volume}{7}}, \bibinfo{pages}{134}
  (\bibinfo{year}{2021}{\natexlab{b}}), ISSN \bibinfo{issn}{2056-6387},
  \urlprefix\url{https://doi.org/10.1038/s41534-021-00474-3}.

\bibitem[{\citenamefont{Zhou et~al.}(2022)\citenamefont{Zhou, Lv, Huang, and
  Ma}}]{zhou2019security}
\bibinfo{author}{\bibfnamefont{H.}~\bibnamefont{Zhou}},
  \bibinfo{author}{\bibfnamefont{K.}~\bibnamefont{Lv}},
  \bibinfo{author}{\bibfnamefont{L.}~\bibnamefont{Huang}}, \bibnamefont{and}
  \bibinfo{author}{\bibfnamefont{X.}~\bibnamefont{Ma}},
  \bibinfo{journal}{IEEE/ACM Transactions on Networking}
  \textbf{\bibinfo{volume}{30}}, \bibinfo{pages}{1328} (\bibinfo{year}{2022}),
  \urlprefix\url{https://doi.org/10.1109/TNET.2021.3136943}.

\bibitem[{\citenamefont{L{\"u}tkenhaus and Ma}(2016)}]{lutkenhaus2016system}
\bibinfo{author}{\bibfnamefont{N.}~\bibnamefont{L{\"u}tkenhaus}}
  \bibnamefont{and} \bibinfo{author}{\bibfnamefont{X.}~\bibnamefont{Ma}},
  \emph{\bibinfo{title}{System and method for quantum key distribution}}
  (\bibinfo{year}{2016}), \bibinfo{note}{uS Patent 9,294,272}.

\bibitem[{\citenamefont{Lo and Chau}(1999)}]{lo1999unconditional}
\bibinfo{author}{\bibfnamefont{H.~K.} \bibnamefont{Lo}} \bibnamefont{and}
  \bibinfo{author}{\bibfnamefont{H.~F.} \bibnamefont{Chau}},
  \bibinfo{journal}{Science} \textbf{\bibinfo{volume}{283}},
  \bibinfo{pages}{2050} (\bibinfo{year}{1999}),
  \urlprefix\url{http://science.sciencemag.org/content/283/5410/2050}.

\bibitem[{\citenamefont{Shor and Preskill}(2000)}]{shor2000simple}
\bibinfo{author}{\bibfnamefont{P.~W.} \bibnamefont{Shor}} \bibnamefont{and}
  \bibinfo{author}{\bibfnamefont{J.}~\bibnamefont{Preskill}},
  \bibinfo{journal}{Phys. Rev. Lett.} \textbf{\bibinfo{volume}{85}},
  \bibinfo{pages}{441} (\bibinfo{year}{2000}),
  \urlprefix\url{https://link.aps.org/doi/10.1103/PhysRevLett.85.441}.

\bibitem[{\citenamefont{Bennett
  et~al.}(1992{\natexlab{a}})\citenamefont{Bennett, Brassard, and
  Mermin}}]{bennett1992quantum}
\bibinfo{author}{\bibfnamefont{C.~H.} \bibnamefont{Bennett}},
  \bibinfo{author}{\bibfnamefont{G.}~\bibnamefont{Brassard}}, \bibnamefont{and}
  \bibinfo{author}{\bibfnamefont{N.~D.} \bibnamefont{Mermin}},
  \bibinfo{journal}{Phys. Rev. Lett.} \textbf{\bibinfo{volume}{68}},
  \bibinfo{pages}{557} (\bibinfo{year}{1992}{\natexlab{a}}),
  \urlprefix\url{https://link.aps.org/doi/10.1103/PhysRevLett.68.557}.

\bibitem[{\citenamefont{Ben-Or et~al.}(2005)\citenamefont{Ben-Or, Horodecki,
  Leung, Mayers, and Oppenheim}}]{ben2005universal}
\bibinfo{author}{\bibfnamefont{M.}~\bibnamefont{Ben-Or}},
  \bibinfo{author}{\bibfnamefont{M.}~\bibnamefont{Horodecki}},
  \bibinfo{author}{\bibfnamefont{D.~W.} \bibnamefont{Leung}},
  \bibinfo{author}{\bibfnamefont{D.}~\bibnamefont{Mayers}}, \bibnamefont{and}
  \bibinfo{author}{\bibfnamefont{J.}~\bibnamefont{Oppenheim}}, in
  \emph{\bibinfo{booktitle}{Proceedings of the Second International Conference
  on Theory of Cryptography}} (\bibinfo{publisher}{Springer-Verlag},
  \bibinfo{address}{Berlin, Heidelberg}, \bibinfo{year}{2005}), TCC'05, pp.
  \bibinfo{pages}{386--406}, ISBN \bibinfo{isbn}{3-540-24573-1,
  978-3-540-24573-5},
  \urlprefix\url{http://dx.doi.org/10.1007/978-3-540-30576-7_21}.

\bibitem[{\citenamefont{Renner and K\"{o}nig}(2005)}]{Renner2005Security}
\bibinfo{author}{\bibfnamefont{R.}~\bibnamefont{Renner}} \bibnamefont{and}
  \bibinfo{author}{\bibfnamefont{R.}~\bibnamefont{K\"{o}nig}}, in
  \emph{\bibinfo{booktitle}{Proceedings of the Second International Conference
  on Theory of Cryptography}} (\bibinfo{publisher}{Springer-Verlag},
  \bibinfo{address}{Berlin, Heidelberg}, \bibinfo{year}{2005}), TCC'05, pp.
  \bibinfo{pages}{407--425}, ISBN \bibinfo{isbn}{3-540-24573-1,
  978-3-540-24573-5},
  \urlprefix\url{http://dx.doi.org/10.1007/978-3-540-30576-7_22}.

\bibitem[{\citenamefont{Bennett et~al.}(1996)\citenamefont{Bennett, DiVincenzo,
  Smolin, and Wootters}}]{Bennett96Mixed}
\bibinfo{author}{\bibfnamefont{C.~H.} \bibnamefont{Bennett}},
  \bibinfo{author}{\bibfnamefont{D.~P.} \bibnamefont{DiVincenzo}},
  \bibinfo{author}{\bibfnamefont{J.~A.} \bibnamefont{Smolin}},
  \bibnamefont{and} \bibinfo{author}{\bibfnamefont{W.~K.}
  \bibnamefont{Wootters}}, \bibinfo{journal}{Phys. Rev. A}
  \textbf{\bibinfo{volume}{54}}, \bibinfo{pages}{3824} (\bibinfo{year}{1996}),
  \urlprefix\url{https://link.aps.org/doi/10.1103/PhysRevA.54.3824}.

\bibitem[{\citenamefont{Ma et~al.}(2011{\natexlab{a}})\citenamefont{Ma, Zhang,
  and Tan}}]{ma2011explicit}
\bibinfo{author}{\bibfnamefont{X.}~\bibnamefont{Ma}},
  \bibinfo{author}{\bibfnamefont{Z.}~\bibnamefont{Zhang}}, \bibnamefont{and}
  \bibinfo{author}{\bibfnamefont{X.}~\bibnamefont{Tan}},
  \bibinfo{journal}{arXiv preprint arXiv:1109.6147}
  (\bibinfo{year}{2011}{\natexlab{a}}),
  \urlprefix\url{https://arxiv.org/abs/1109.6147}.

\bibitem[{\citenamefont{Yang et~al.}(2019)\citenamefont{Yang, Horodecki, and
  Winter}}]{yang2019distributed}
\bibinfo{author}{\bibfnamefont{D.}~\bibnamefont{Yang}},
  \bibinfo{author}{\bibfnamefont{K.}~\bibnamefont{Horodecki}},
  \bibnamefont{and} \bibinfo{author}{\bibfnamefont{A.}~\bibnamefont{Winter}},
  \bibinfo{journal}{Phys. Rev. Lett.} \textbf{\bibinfo{volume}{123}},
  \bibinfo{pages}{170501} (\bibinfo{year}{2019}),
  \urlprefix\url{https://link.aps.org/doi/10.1103/PhysRevLett.123.170501}.

\bibitem[{\citenamefont{Fung et~al.}(2012)\citenamefont{Fung, Ma, Chau, and
  Cai}}]{fung2012quantum}
\bibinfo{author}{\bibfnamefont{C.-H.~F.} \bibnamefont{Fung}},
  \bibinfo{author}{\bibfnamefont{X.}~\bibnamefont{Ma}},
  \bibinfo{author}{\bibfnamefont{H.~F.} \bibnamefont{Chau}}, \bibnamefont{and}
  \bibinfo{author}{\bibfnamefont{Q.-y.} \bibnamefont{Cai}},
  \bibinfo{journal}{Phys. Rev. A} \textbf{\bibinfo{volume}{85}},
  \bibinfo{pages}{032308} (\bibinfo{year}{2012}),
  \urlprefix\url{https://link.aps.org/doi/10.1103/PhysRevA.85.032308}.

\bibitem[{\citenamefont{Stacey et~al.}(2015)\citenamefont{Stacey, Annabestani,
  Ma, and L\"utkenhaus}}]{Stacey2015Relay}
\bibinfo{author}{\bibfnamefont{W.}~\bibnamefont{Stacey}},
  \bibinfo{author}{\bibfnamefont{R.}~\bibnamefont{Annabestani}},
  \bibinfo{author}{\bibfnamefont{X.}~\bibnamefont{Ma}}, \bibnamefont{and}
  \bibinfo{author}{\bibfnamefont{N.}~\bibnamefont{L\"utkenhaus}},
  \bibinfo{journal}{Phys. Rev. A} \textbf{\bibinfo{volume}{91}},
  \bibinfo{pages}{012338} (\bibinfo{year}{2015}),
  \urlprefix\url{http://link.aps.org/doi/10.1103/PhysRevA.91.012338}.

\bibitem[{\citenamefont{Alon et~al.}(1992)\citenamefont{Alon, Goldreich,
  H{\aa}stad, and Peralta}}]{alon1992simple}
\bibinfo{author}{\bibfnamefont{N.}~\bibnamefont{Alon}},
  \bibinfo{author}{\bibfnamefont{O.}~\bibnamefont{Goldreich}},
  \bibinfo{author}{\bibfnamefont{J.}~\bibnamefont{H{\aa}stad}},
  \bibnamefont{and} \bibinfo{author}{\bibfnamefont{R.}~\bibnamefont{Peralta}},
  \bibinfo{journal}{Random Structures \& Algorithms}
  \textbf{\bibinfo{volume}{3}}, \bibinfo{pages}{289} (\bibinfo{year}{1992}),
  \urlprefix\url{https://doi.org/10.1002/rsa.3240030308}.

\bibitem[{\citenamefont{Streltsov et~al.}(2017)\citenamefont{Streltsov, Adesso,
  and Plenio}}]{streltsov2017colloquium}
\bibinfo{author}{\bibfnamefont{A.}~\bibnamefont{Streltsov}},
  \bibinfo{author}{\bibfnamefont{G.}~\bibnamefont{Adesso}}, \bibnamefont{and}
  \bibinfo{author}{\bibfnamefont{M.~B.} \bibnamefont{Plenio}},
  \bibinfo{journal}{Rev. Mod. Phys.} \textbf{\bibinfo{volume}{89}},
  \bibinfo{pages}{041003} (\bibinfo{year}{2017}),
  \urlprefix\url{https://link.aps.org/doi/10.1103/RevModPhys.89.041003}.

\bibitem[{\citenamefont{Ma et~al.}(2016)\citenamefont{Ma, Yuan, Cao, Qi, and
  Zhang}}]{ma2016quantum}
\bibinfo{author}{\bibfnamefont{X.}~\bibnamefont{Ma}},
  \bibinfo{author}{\bibfnamefont{X.}~\bibnamefont{Yuan}},
  \bibinfo{author}{\bibfnamefont{Z.}~\bibnamefont{Cao}},
  \bibinfo{author}{\bibfnamefont{B.}~\bibnamefont{Qi}}, \bibnamefont{and}
  \bibinfo{author}{\bibfnamefont{Z.}~\bibnamefont{Zhang}},
  \bibinfo{journal}{npj Quantum Information} \textbf{\bibinfo{volume}{2}},
  \bibinfo{pages}{1} (\bibinfo{year}{2016}),
  \urlprefix\url{https://www.nature.com/articles/npjqi201621}.

\bibitem[{\citenamefont{Tsurumaru}(2022)}]{Tsurumaru2022equivalence}
\bibinfo{author}{\bibfnamefont{T.}~\bibnamefont{Tsurumaru}},
  \bibinfo{journal}{{IEEE} Transactions on Information Theory}
  \textbf{\bibinfo{volume}{68}}, \bibinfo{pages}{1016} (\bibinfo{year}{2022}),
  \urlprefix\url{https://doi.org/10.1109%2Ftit.2021.3126160}.

\bibitem[{\citenamefont{Koashi}(2009)}]{koashi2009simple}
\bibinfo{author}{\bibfnamefont{M.}~\bibnamefont{Koashi}}, \bibinfo{journal}{New
  J. Phys.} \textbf{\bibinfo{volume}{11}}, \bibinfo{pages}{045018}
  (\bibinfo{year}{2009}),
  \urlprefix\url{https://doi.org/10.1088/1367-2630/11/4/045018}.

\bibitem[{\citenamefont{Gottesman et~al.}(2004)\citenamefont{Gottesman, Lo,
  L\"{u}tkenhaus, and Preskill}}]{gottesman2004security}
\bibinfo{author}{\bibfnamefont{D.}~\bibnamefont{Gottesman}},
  \bibinfo{author}{\bibfnamefont{H.-K.} \bibnamefont{Lo}},
  \bibinfo{author}{\bibfnamefont{N.}~\bibnamefont{L\"{u}tkenhaus}},
  \bibnamefont{and} \bibinfo{author}{\bibfnamefont{J.}~\bibnamefont{Preskill}},
  \bibinfo{journal}{Quantum Info. Comput.} \textbf{\bibinfo{volume}{4}},
  \bibinfo{pages}{325} (\bibinfo{year}{2004}), ISSN \bibinfo{issn}{1533-7146},
  \urlprefix\url{http://dl.acm.org/citation.cfm?id=2011586.2011587}.

\bibitem[{\citenamefont{Bennett
  et~al.}(1992{\natexlab{b}})\citenamefont{Bennett, Brassard, Cr{\'e}peau, and
  Skubiszewska}}]{bennett1992practical}
\bibinfo{author}{\bibfnamefont{C.~H.} \bibnamefont{Bennett}},
  \bibinfo{author}{\bibfnamefont{G.}~\bibnamefont{Brassard}},
  \bibinfo{author}{\bibfnamefont{C.}~\bibnamefont{Cr{\'e}peau}},
  \bibnamefont{and} \bibinfo{author}{\bibfnamefont{M.-H.}
  \bibnamefont{Skubiszewska}}, in \emph{\bibinfo{booktitle}{Advances in
  Cryptology --- CRYPTO '91}}, edited by
  \bibinfo{editor}{\bibfnamefont{J.}~\bibnamefont{Feigenbaum}}
  (\bibinfo{publisher}{Springer Berlin Heidelberg}, \bibinfo{address}{Berlin,
  Heidelberg}, \bibinfo{year}{1992}{\natexlab{b}}), pp.
  \bibinfo{pages}{351--366},
  \urlprefix\url{https://link.springer.com/chapter/10.1007/3-540-46766-1_29}.

\bibitem[{\citenamefont{Wei et~al.}(2017)\citenamefont{Wei, Cai, Liu, Wang, and
  Gao}}]{wei2017generic}
\bibinfo{author}{\bibfnamefont{C.-Y.} \bibnamefont{Wei}},
  \bibinfo{author}{\bibfnamefont{X.-Q.} \bibnamefont{Cai}},
  \bibinfo{author}{\bibfnamefont{B.}~\bibnamefont{Liu}},
  \bibinfo{author}{\bibfnamefont{T.-Y.} \bibnamefont{Wang}}, \bibnamefont{and}
  \bibinfo{author}{\bibfnamefont{F.}~\bibnamefont{Gao}}, \bibinfo{journal}{IEEE
  Trans Comput} \textbf{\bibinfo{volume}{67}}, \bibinfo{pages}{2}
  (\bibinfo{year}{2017}),
  \urlprefix\url{https://doi.org/10.1109/TC.2017.2721404}.

\bibitem[{\citenamefont{Renner}(2008)}]{renner2008security}
\bibinfo{author}{\bibfnamefont{R.}~\bibnamefont{Renner}},
  \bibinfo{journal}{Int. J. Quantum Inf.} \textbf{\bibinfo{volume}{6}},
  \bibinfo{pages}{1} (\bibinfo{year}{2008}),
  \urlprefix\url{https://doi.org/10.1142/S0219749908003256}.

\bibitem[{\citenamefont{Brassard and Salvail}(1994)}]{brassard1994advances}
\bibinfo{author}{\bibfnamefont{G.}~\bibnamefont{Brassard}} \bibnamefont{and}
  \bibinfo{author}{\bibfnamefont{L.}~\bibnamefont{Salvail}},
  \bibinfo{journal}{Lecture Notes in Computer Science}
  \textbf{\bibinfo{volume}{765}}, \bibinfo{pages}{410} (\bibinfo{year}{1994}).

\bibitem[{\citenamefont{L\"utkenhaus}(1999)}]{lutkenhaus1999estimates}
\bibinfo{author}{\bibfnamefont{N.}~\bibnamefont{L\"utkenhaus}},
  \bibinfo{journal}{Phys. Rev. A} \textbf{\bibinfo{volume}{59}},
  \bibinfo{pages}{3301} (\bibinfo{year}{1999}),
  \urlprefix\url{https://link.aps.org/doi/10.1103/PhysRevA.59.3301}.

\bibitem[{\citenamefont{Ma et~al.}(2011{\natexlab{b}})\citenamefont{Ma, Fung,
  Boileau, and Chau}}]{ma2011universally}
\bibinfo{author}{\bibfnamefont{X.}~\bibnamefont{Ma}},
  \bibinfo{author}{\bibfnamefont{C.-H.~F.} \bibnamefont{Fung}},
  \bibinfo{author}{\bibfnamefont{J.-C.} \bibnamefont{Boileau}},
  \bibnamefont{and} \bibinfo{author}{\bibfnamefont{H.}~\bibnamefont{Chau}},
  \bibinfo{journal}{Comput Secur} \textbf{\bibinfo{volume}{30}},
  \bibinfo{pages}{172} (\bibinfo{year}{2011}{\natexlab{b}}), ISSN
  \bibinfo{issn}{0167-4048},
  \urlprefix\url{https://www.sciencedirect.com/science/article/pii/S0167404810001021}.

\bibitem[{\citenamefont{Lo}(2003)}]{lo2003method}
\bibinfo{author}{\bibfnamefont{H.-K.} \bibnamefont{Lo}}, \bibinfo{journal}{New
  J. Phys.} \textbf{\bibinfo{volume}{5}}, \bibinfo{pages}{36}
  (\bibinfo{year}{2003}),
  \urlprefix\url{https://doi.org/10.1088/1367-2630/5/1/336}.

\bibitem[{\citenamefont{Gallager}(1962)}]{gallager1962low}
\bibinfo{author}{\bibfnamefont{R.}~\bibnamefont{Gallager}},
  \bibinfo{journal}{IRE Trans. Inf. Theory} \textbf{\bibinfo{volume}{8}},
  \bibinfo{pages}{21} (\bibinfo{year}{1962}),
  \urlprefix\url{https://doi.org/10.1109/TIT.1962.1057683}.

\bibitem[{\citenamefont{MacKay}(1999)}]{mackay1999good}
\bibinfo{author}{\bibfnamefont{D.~J.} \bibnamefont{MacKay}},
  \bibinfo{journal}{IEEE Trans. Inf. Theory} \textbf{\bibinfo{volume}{45}},
  \bibinfo{pages}{399} (\bibinfo{year}{1999}),
  \urlprefix\url{https://doi.org/10.1109/18.748992}.

\bibitem[{\citenamefont{Ma}(2008)}]{Ma2008PhD}
\bibinfo{author}{\bibfnamefont{X.}~\bibnamefont{Ma}}, Ph.D. thesis,
  \bibinfo{school}{University of Toronto} (\bibinfo{year}{2008}),
  \bibinfo{note}{also available in arXiv:0808.1385},
  \urlprefix\url{https://arxiv.org/abs/0808.1385}.

\end{thebibliography}

\end{document}